%% file: part_one_arxiv.tex
\documentclass{scrartcl}

\usepackage{amsmath}
\usepackage{xargs}
\usepackage{ifthen}
\usepackage{dsfont}
\usepackage[disable]{todonotes}
\usepackage{amssymb}
\usepackage{bm} 
\usepackage[process=all]{pstool}
\usepackage{stmaryrd}
\usepackage{subfig}

\newcommand{\mat}{\mathbf}
\newcommand{\wt}[1]{\widetilde{#1}}
\newcommand{\N}{\mathds{N}}
\newcommand{\R}{\mathds{R}}

\newcommand{\C}{\mathds{C}}
\newcommand{\UF}{\mathcal{U}}
\newcommand{\F}{\mathfrak{F}}

\newcommand{\eps}{\varepsilon}
\newcommand{\ND}{\mathcal{N}}
\newcommand{\set}[1]{\left\{#1\right\}}
\newcommand{\abs}[1]{\left|#1\right|}
\newcommand{\norm}[1]{\left\|#1\right\|}
\newcommand{\To}{\mathop{\rightarrow}}
\newcommandx{\E}[2][2=\empty]{\ifthenelse{\equal{#2}{\empty}}{\mathrm{E}}{\mathrm{E}_{#2}}\!\left(#1\right)}
\newcommand{\Var}[1]{\mathrm{Var}\left(#1\right)}
\newcommand{\Prob}[1]{\mathrm{Pr}\!\left(#1\right)}

\newcommand{\psfragsize}{\small}

\newcommandx{\psf}[3][3=c]{\psfrag{#1}[#3][]{$\psfragsize{#2}$}}
\newcommandx{\psft}[3][3=c]{\psfrag{#1}[#3][]{\psfragsize{#2}}}

\renewcommand{\vec}{\bm}
\newcommand{\supp}{\text{supp}}
\newcommand{\hrnorm}[1]{\llbracket #1\rrbracket_\infty}

\newtheorem{thm}{Theorem}[section]{\bfseries}{\rmfamily}
\newtheorem{lem}[thm]{Lemma}{\bfseries}{\rmfamily}
\newtheorem{prop}[thm]{Proposition}{\bfseries}{\rmfamily}
{\bfseries}{\rmfamily}
\newtheorem{defn}[thm]{Definition}{\bfseries}{\rmfamily}
\newtheorem{exa}[thm]{Example}{\bfseries}{\rmfamily}
\newtheorem{rem}[thm]{Remark}{\itshape}{\rmfamily}

\newtheorem{assumption}[thm]{Assumption}{\bfseries}{\rmfamily}
{\bfseries}{\rmfamily}

\newenvironmentx{proof}[1][1=\empty]{\ifthenelse{\equal{#1}{\empty}}{\emph{Proof.}~}{\emph{Proof (#1).~}}}{\hfill$\square$\\[0.1\baselineskip]}

\title{On Game-Theoretic Risk Management (Part One)}
\subtitle{Towards a Theory of Games with Payoffs that are Probability-Distributions}
\author{Stefan Rass
\thanks{Universit\"{a}t Klagenfurt, Institute of Applied Informatics,
System Security Group, Universit\"{a}tsstrasse 65-67, 9020 Klagenfurt, Austria.
This work has been done in the course of consultancy for the EU Project
HyRiM (Hybrid Risk Management for Utility Networks; see https://hyrim.net),
led by the \emph{Austrian
Institute of Technology} (AIT; www.ait.ac.at). See the acknowledgement section.}\\
\texttt{stefan.rass@aau.at}}

\begin{document}

\maketitle

\begin{abstract}
\begin{center}
\textbf{Abstract}
\end{center}
 Optimal behavior in (competitive) situation is traditionally determined
with the help of utility functions that measure the payoff of different
actions. Given an ordering on the space of revenues (payoffs), the
classical axiomatic approach of von Neumann and Morgenstern establishes the
existence of suitable utility functions, and yields to game-theory as the
most prominent materialization of a theory to determine optimal behavior.
Although this appears to be a most natural approach to risk management too,
applications in critical infrastructures often violate the implicit
assumption of actions leading to deterministic consequences. In that sense,
the gameplay in a critical infrastructure risk control competition is
intrinsically random in the sense of actions having uncertain consequences.
Mathematically, this takes us to utility functions that are
probability-distribution-valued, in which case we loose the canonic (in
fact every possible) ordering on the space of payoffs, and the original
techniques of von Neumann and Morgenstern no longer apply.

This work introduces a new kind of game in which uncertainty applies to the
payoff functions rather than the player's actions (a setting that has been
widely studied in the literature, yielding to celebrated notions like the
trembling hands equilibrium or the purification theorem). In detail, we
show how to fix the non-existence of a (canonic) ordering on the space of
probability distributions by only mildly restricting the full set to a
subset that can be totally ordered. Our vehicle to define the ordering and
establish basic game-theory is non-standard analysis and hyperreal numbers.
\end{abstract}

\tableofcontents \newpage
\input{report-Y1_draft.tex}

\section*{Acknowledgment}
This work was supported by the European Commission's Project No. 608090,
HyRiM (Hybrid Risk Management for Utility Networks) under the 7th Framework
Programme (FP7-SEC-2013-1). The author is indepted to Sandra K\"{o}nig from
AIT for valuable discussions and for spotting some errors. The author also
thanks Vincent Bürgin for pointing to some subtleties in the arguments of
Lemma \ref{lem:ordering-invariance}, which led to a refinement of it and its
implied results, and a valuable clarification in Remark
\ref{rem:responding-to-vincent}.


\input{part_one_arxiv_bib.bbl}
\end{document}

%% file: report-Y1_draft.tex
\section{Introduction}
Security risk management is a continuous cycle of action and reaction to the
changing working conditions of an infrastructure. This cycle is detailed in
relevant standards like ISO 2700x, where phases designated to
\emph{planning}, \emph{doing}, \emph{checking} and \emph{acting} are
rigorously defined and respective measures are given.

Our concern in this report is an investigation of a hidden assumption
underneath this recommendation, namely the hypothesis that some wanted impact
can be achieved by taking the proper action. If so, then security risk
management would degenerate to a highly complex but nevertheless
deterministic control problem, to which optimal solutions and strategies
could be found (at least in theory).

Unfortunately, however, reality is intrinsically random to some extent, and
the outcome of an action is almost never fully certain. Illustrative examples
relate to how public opinion and trust are dependent on the public relation
strategies of an institution. While there are surely ways to influence the
public opinion, it will always be ultimately out of one's full and exclusive
control. Regardless of this, we ought to find optimal ways to influence the
situation in the way we like. This can -- in theory -- again be boiled down
to a (not so simple) optimization problem, however, one that works on
optimizing partially random outcomes. This is where things start to get
nontrivial.

Difficulties in the defense against threats root in the nature of relevant
attacks, since not all of them are immediately observable or induce instantly
noticeable or measurable consequences. Indeed, the best we can do is finding
an optimal protection against an a-priori identified set of attack scenarios,
so as to gain the assurance of security against the known list of threat
scenarios. Optimizing this protection is often, but not necessarily, tied to
some kind of \emph{adversary modelling}, in an attempt to sharpen our
expectations about what may happen to us. Such adversary modeling is
inevitably error prone, as the motives and incentives of an attacker may
deviate from our imagination to an arbitrary extent.

Approaching the problem mathematically, there are two major lines of decision
making: one works with an a-priori hypothesis of the current situation, and
incorporates current information into an a-posteriori model that tells us how
things will evolve, and specifically, which events are more likely than
others, given the full information that we have. Decision making in that case
means that we seek the optimal behavior so as to master a \emph{specifically}
expected setting (described by the a-posteriori distribution). This is the
\emph{Bayesian} approach to decision making (see \cite{Robert2001} for a
fully comprehensive detailed). The second way of decision making explicitly
avoids any hypothesis about the current situation, and seeks an optimal
behavior against \emph{any} possible setting. Unlike the Bayesian
perspective, we would thus intentionally and completely ignore all available
data and choose our actions to master the worst-case scenario. While this
so-called \emph{minimax decision making} is obviously a more pessimistic and
cautious approach, it appears better suited for risk management in situations
where data is either not available, not trustworthy or inaccurate.

For this reason, we will hereafter pursue the minimax-approach and dedicate
section \ref{sec:relation-to-bayesian-choice} to a discussion how this fits
into the Bayesian framework as a special case.

We assume that the risk manager can repeatedly take actions and that the
possible actions are finitely many. Furthermore, we assume that the adversary
against which we do our risk control, also has a finite number of possible
ways to cause trouble. In terms of an ISO 2700k risk management process, the
risk manager's actions would instantiate \emph{controls}, while the
adversary's actions would correspond to identified threat scenarios. The
assumption of finiteness does stringently constrain us here, as an infinite
number of actions to choose from may in any case overstrain a human
decision-maker.

The crucial point in all that follows is that any action (as taken by the
decision maker) in any situation (action taken by the adversary) may have an
intended but in any case random outcome. To properly formalize this and fit
it into a mathematical, in fact game-theoretic, framework, we hereafter
associate the risk manager with \emph{player 1} in our game, who competes
with \emph{player 2}, who is the adversary. Actions of either players are in
game-theoretic literature referred to as \emph{pure strategies}; the entirety
of which will be abbreviated as $PS_1$ and $PS_2$ for either player, so
$PS_1$ comprises all actions, hereafter called \emph{strategies} available
for risk management, while $PS_2$ comprises all trouble scenarios. For our
treatment, it is not required to be specific on how the elements in both
action sets look like, as it is sufficient for them to ``be available''.

Let $PS_1, PS_2$ denote finite sets of strategies for two players, where
player 1 is the honest defender (e.g., utility infrastructure provider), and
player 2 is the adversary. We assume player 1 to be unaware of its opponents
incentives, so that an optimal strategy is sought against any possible
behavior within the \emph{known} action space $PS_2$ of the opponent
(rational or irrational, e.g., nature),.

In this sense, $PS_2$ can be the set of all known possible security
incidents, whose particular incarnations can become reality by the
adversary's action. To guard its assets, player 1 can choose from a finite
set of actions $PS_1$ to minimize the costs of a recovery from any incident,
or equivalently, keep its risk under control.

Upon these assumptions, the situation can be described by an $(n\times
m)$-matrix of scenarios, where $n=\abs{PS_1}, m = \abs{PS_2}$, each of which
is associated with some cost $R_{ij}$ to recover the system from a
malfunctioning state back to normal operation from scenario $(i,j)\in
PS_1\times PS_2$. We use the variable $R_{ij}$ henceforth to denote the cost
of a \emph{repair} made necessary by an incident $j\in PS_2$ happening when
the system is currently in configuration $i\in PS_1$.

The process of risk management will be associated with player 1 putting the
system into different configurations over time in order to minimize the risk
$R_{ij}$.

\begin{rem} We leave the exact understanding of ``risk'' or ``damage'' intentionally
undefined here, as this will be quite different between various utility
infrastructures or general fields of application.
\end{rem}

\begin{rem} Neither the set $PS_1$ nor the set $PS_2$ is here specified in
any detail further than declaring it as an ``action space''. The reason is,
again, the expected diversity of actions and incidents among various fields
of application (or utility infrastructures). Therefore and to keep this
report as general and not limiting the applicability of the results to
follow, we will leave the precise elements of $PS_1, PS_2$ up to definitions
that are tailored to the intended application.

Examples of strategies may include:
\begin{itemize}
  \item random spot checks in the system to locate and fix problems
      (ultimately, to keep the system running),
  \item random surveillance checks and certain locations,
  \item certain efforts or decisions about whether or not, and which, risks
      or countermeasures shall be communicated to the society or user
      community,
  \item etc.
\end{itemize}
\end{rem}

In real life settings, it can be expected that an action (regardless of who
takes it), always has some intrinsic randomness. That is, the effect of a
particular scenario $(i,j)\in PS_1\times PS_2$ is actually a random variable
$R_{ij}$, having only some ``expected'' outcome that may be different between
any two occurrences of the same situation $(i,j)$ over time.

To be able to properly handle the arising random variables, let us think of
those modeling not the benefits but rather the \emph{damage} that a security
incident may cause. In this view, we can go for minimization of an expectedly
positive value that measures the cost of a recovery. Formally, we introduce
the following assumption that will greatly ease theoretical technicalities
throughout this work, while not limiting the practicability too much.

The family $\set{R_{ij}:i\in PS_1, j\in PS_2}$ of random damage distributions
in our game will be assumed with all members satisfying the following
assumption:
\begin{assumption}\label{asm:finity-of-repair-costs}
Let $R_{ij}$ be a real-valued random variable. On $R_{ij}$, we impose the
following assumptions:
\begin{itemize}
  \item $R_{ij}\geq 1$ (w.l.o.g.\footnote{It is common to assume losses to
      be $\geq 0$; our modification has technical reasons, but causes no
      semantic difference in the comparisons between two loss densities,
      since both loss variables are just shifted by the same amount. Also,
      the loss can (w.l.o.g.) be scaled until losses in the range $[0,1)$
      become practically negligible.}).
  \item $R_{ij}$ has a known distribution $F_{ij}$ with compact support
      (note that this implies that all $R_{ij}$ is upper-bounded).
  \item The probability measure induced by $F_{ij}$ is either discrete or
      continuous and has a density function $f_{ij}$. For continuous random
      variables, the density function is assumed to be continuous.
\end{itemize}
\end{assumption}

\subsection{Symbols and Notation}
\input{notation}

\section{Optimal Decisions under Uncertainty}\label{sec:preferences}
Under the above setting, we can collect all scenarios of actions that player
1 (defender) and player 2 (attacker) may take in a tabular (matrix) fashion.
Our goal in this first step, is to soundly define what ``a best action''
would be in light of uncertain, indeed random, effects that actions on either
side cause, especially in lack of control about the other's actions. For that
matter, we will consider the scenario matrix $\mat A$ as given below, as the
payoff structure of some matrix-game, whose mathematical underpinning is the
standard setup of game-theory (see \cite{Fudenberg1991} for example), with
differences and necessary changes to classical theory of games, being
discussed in sections \ref{sec:generalized-game-theory} and later.

Let the following tableau be a collection of all scenarios of actions taken
by the defender (row-player) and attacker (column-player),
\[
\mat A=\left(
  \begin{array}{ccccc}
    R_{11} & \cdots & R_{1j} & \cdots & R_{1m} \\
    \vdots & \ddots & \vdots & \ddots & \vdots \\
    R_{i1} & \cdots & R_{ij} & \cdots & R_{im} \\
    \vdots & \ddots & \vdots & \ddots & \vdots \\
    R_{n1} & \cdots & R_{nj} & \cdots & R_{nm} \\
  \end{array}
\right),
\]
where the rows of the matrix $\mat A$ are labeled by the actions in $PS_1$,
and the columns of $\mat A$ carry the labels of actions from $PS_2$.

A \emph{security strategy} for player 1 is an optimal choice $i^*$ of a row
so that the risk, expressed by the random variable $R_{i^*j}$ is
``optimized'' over all possible actions $j\in PS_2$ of the opponent. Here, we
run into trouble already, as there is no canonical ordering on the set of
probability distributions.

To the end of resolving this issue, let us consider repetitions of the
gameplay in which each player can choose his actions repeatedly and
differently, in an attempt to minimize risk (or damage). This corresponds to
situations in which ``the best'' configuration simply does not exist, and we
are forced to repeatedly change or reconsider the configuration of the system
in order to remain protected.

In a classical game-theoretic approach, this takes us to the concept of
\emph{mixed strategies}, which are discrete probability distributions over
the action spaces of the players. Making this rigorous, let $S(PS_i)$ for
$i=1,2$ denote the simplex over $PS_i$, i.e., the space of all discrete
probability distributions supported on $PS_i$. More formally, given the
support $X$, the set $S(X)$ is
\[
    S(X) := \set{(p_1,\ldots,p_k)\in\R^k: k=\abs{X}, \sum_{i=1}^k p_i = 1, p_i\geq 0\,\forall i}.
\]
A randomized decision is thus a rule $\vec p=(p_1,\ldots,p_k)\in S(PS)$ to
choose from the available actions $\set{1,2,\ldots,k}$ from the action set
$PS$ (which is $PS_1$ or $PS_2$ hereafter) with corresponding probabilities
$p_i$. We assume the ordering of the actions to be arbitrary but fixed (for
obvious reasons).

Now, we return to the problem of what effect to expect when the current
configuration of the system is randomly drawn from $PS_1$, and the
adversary's action is another random choice from $PS_2$. For that matter, let
us simplify notation by putting $S_1 := S(PS_1), S_2 := S(PS_2)$ and let the
two mixed strategies be $\vec p\in S_1$ for player 1, and $\vec q\in S_2$ for
player 2.

Since the choice from the matrix $\mat A$ is random, where the row is drawn
with likelihoods as specified by $\vec p$, and the column is drawn from $\vec
q$, the law of total probability yields for the outcome $R$,
\begin{equation}\label{eqn:game-outcome}
    \Prob{R\leq r} = \sum_{i,j}\Prob{R_{ij}\leq r|i,j}\Prob{i,j},
\end{equation}
where $\Prob{R_{ij}\leq r|i,j}$ is the conditional probability of $R_{ij}$
given a particular choice $(i,j)$, and $\Prob{i,j}$ is the (unconditional)
probability for this choice to occur. Section \ref{sec:dependent-actions}
gives some more details on how $\Prob{i,j}$ can be modeled and expressed.

Denote by $F(\vec p,\vec q)$ the distribution of the game's outcome under
strategies $(\vec p,\vec q)\in S_1\times S_2$, then $\Prob{R\leq r}=F(r)$
depends on $(\vec p,\vec q)$, and \eqref{eqn:game-outcome} can be rewritten
as
\begin{equation}\label{eqn:game-outcome-distribution}
    \Prob{R\leq r}=(F(\vec p,\vec q))(r)=\sum_{i,j}F_{ij}(r)C_{\vec p,\vec q}(i,j),
\end{equation}
where $C_{\vec p,\vec q}(i,j)=\Prob{i,j}$ will be assumed as continuous in
$\vec p$ and $\vec q$ for technical reasons that will become evident later
(during the proof of proposition \ref{prop:continuity}). Note that the
distribution $F$ via the function $C$ explicitly depends on the choices $\vec
p,\vec q$, and is to be ``optimally shaped'' w.r.t. to these two variables.
The argument $r\in\R$ to the function $F(\vec p,\vec q)(\cdot)$ is the
(random) ``revenue'', whose uncertainty is outside any of the two player's
influence (besides shaping $F$ by proper choices of $\vec p$ and $\vec q$).

The ``revenue'' $R$ in the game can be of manifold nature, such as
\begin{itemize}
  \item Risk \underline{r}esponse of society; a quantitative measure that
      could rate people's opinions and confidence in the utility
      infrastructure
  \item \underline{R}epair cost to recover from an incident's implied
      damage,
  \item \underline{R}eliability, if the game is about whether or not a
      particular quality of service can be kept up,
  \item etc.
\end{itemize}

\begin{rem} In the simplest case of
independent actions, we would set $C_{\vec p,\vec q}(i,j)=p_i\cdot q_j$ when
$\vec p=(p_1,\ldots,p_n),\vec q = (q_1,\ldots, q_m)$. This choice, along with
assuming $R_{ij}$ to be constants rather than random variables, recreates the
familiar matrix-game payoff functional $\vec p^T\mat A\vec q$ from
\eqref{eqn:game-outcome-distribution}. Hence,
\eqref{eqn:game-outcome-distribution} is a first generalization of matrix
games to games with uncertain outcome, which for the sake of flexibility and
generality, is ``distribution-valued''.
\end{rem}

\begin{rem} It may be in reality the case that actions of the two players are
\emph{not} chosen independently, for example, if both of the players possess
some common knowledge or access to a common source of information. In
game-theoretic terms, this would lead to so-called \emph{correlated
equilibria} (see \cite{Fudenberg1991}), in which the players share two
correlated random variables that influence their choices. Things here are
nevertheless different, as no bidirectional flow of information can be
assumed like for correlated equilibria (the attacker won't inform the utility
infrastructure provider about anything in advance, while information from the
provider may somehow leak out to the adversary).
\end{rem}

\subsection{Choosing Actions (In)dependently}\label{sec:dependent-actions}

The concrete choice of the function $C_{\vec p,\vec q}$ is only subject to
continuity in $\vec p,\vec q$ for technical reasons that will receive a
closer look now. The general joint probability of the scenario $(i,j)$ w.r.t.
the marginal discrete distribution vectors $\vec p,\vec q$ is $\Pr_{\vec
p,\vec q}\set{i,j}=\Pr_{\vec p,\vec q}\set{X=i,Y=j}=C_{\vec p,\vec q}(i,j)$
in \eqref{eqn:game-outcome}. Under independence of the random choices
$X\sim\vec p, Y\sim \vec q$ can be written as $\Prob{i,j} =
\Prob{X=i}\Prob{Y=j}=p_iq_j$.

Now, let us consider cases where the choices are \emph{not} independent, say,
if one player observes the other player's actions and can react on them (or
if both players have access to common source of information).

Sklar's theorem implies the existence of a copula-function $C$ so that the
joint distribution $F_{(X,Y)}$ can be written in terms of the copula $C$ and
the marginal distributions $F_X$, corresponding to the vector $\vec p$, and
$F_Y$, corresponding to the vector $\vec q$,
\[
    F_{(X,Y)}(i,j) = \Prob{X\leq i,Y\leq j} = C(F_X(i), F_Y(j)).
\]
\begin{align}
  \Prob{i,j} &= \Prob{X=i,Y=i} = \Prob{X\leq i,Y\leq j} - \Prob{X\leq i-1,Y\leq j}\nonumber \\
&\qquad\qquad- \Prob{X\leq i,Y\leq j-1}+\Prob{X\leq i-1,Y\leq j-1}\nonumber\\
&= C(F_X(i),F_Y(j)) - C(F_X(i-1),F_Y(j)) \nonumber\\
&\qquad\qquad- C(F_X(i), F_Y(j-1)) + C(F_X(i-1),F_Y(j-1))\nonumber\\
&= C(p_i,q_j) - C(p_{i-1},q_j) - C(p_i,q_{j-1}) + C(p_{i-1},q_{j-1})\label{eqn:copula-representation}.
\end{align}
Thus, the function $C_{\vec p,\vec q}$ can be constructed from
\eqref{eqn:copula-representation} based on the copula $C$ (which must exist).
Continuity of $C_{\vec p,\vec q}$ thus hinges on the continuity of the copula
function. At least two situations admit a choice of $C$ that makes $C_{\vec
p,\vec q}$ continuous:
\begin{itemize}
  \item Independence of actions: $C(x,y) := x\cdot y$
  \item Complete lack of knowledge about the interplay between the action
      choices, in which case we can set $C(x,y):=\min\set{x,y}$.

This choice is justified upon the well-known Fr\'{e}chet-Hoeffding bound,
which says that \emph{every} $n$-dimensional copula function $C$ satisfies
\[
C(u_1,u_2,\ldots,u_n)\leq\min\set{u_1,\ldots,u_n}.
\]
Since the $\min$-function is itself a copula, it can be chosen if a
dependency is known to exist, but with no details on the particular nature
of the interplay. Observe that this corresponds to the well-known
\emph{maximum-principle of system security}, where the overall system risk
is determined from the maximum risk among its components (alternatively,
you may think of a chain to be as strong as its weakest element; which
corresponds to the $\min$-function among all indicators $u_1,\ldots,u_n$).
\end{itemize}

\subsection{Comparing Payoff Distributions}

There appears to be no canonical way to compare payoff distributions, as
$F(\vec p,\vec q)$ can be determined by an arbitrary number of parameters,
thus introducing ambiguity in how to compare them. To see this, simply
consider the set of normal distributions $\ND(\mu,\sigma^2)$ being determined
by two parameters $\mu$ and $\sigma>0$. Since the pair $(\mu,\sigma)$
uniquely determines the distribution function, a comparison between two
members $F_1,F_2\in\ND$ amounts to a criterion to compare two-dimensional
vectors $(\mu,\sigma)\in\R\times\R^+\subset\R^2$. It is well-known that
$\R^2$ is not ordered (as being isomorphic to $\C$, on which provably no
order exists; see \cite{Endl1995} for a proof), and hence there is no natural
ordering on the set of probability distributions either.

Despite this sounding like bad news, we can actually construct an alternative
characterization of probability distributions on a new space, in which the
distributions of interest, in our case $F(\vec p,\vec q)$ will all be members
of a totally ordered subset.


To this end, we will rely on a characterization of a probability distribution
of the random variable $R\sim F(\vec p,\vec q)$ via the sequence
$(m_R(k))_{k\in\N}$ of its moments. The $k$-th such moment is from
\eqref{eqn:game-outcome-distribution} and by assumption
\ref{asm:finity-of-repair-costs} found to be
\begin{align}
    [\E{R^k}](\vec p,\vec q) &= \int_{-\infty}^\infty x^kdF(\vec p,\vec q)=\int_{-\infty}^{\infty} x^k\sum_{i,j}f_{ij}(x)C_{\vec p,\vec q}(i,j)dx\nonumber\\
    &=\sum_{i,j}C_{\vec p,\vec q}(i,j)\int_{-\infty}^{\infty} x^kf_{ij}(x)dx=\sum_{i,j}C_{\vec p,\vec q}(i,j)\E{R_{ij}^k},\label{eqn:moment}
\end{align}
where the sum runs over $i=1,2,\ldots,n$ and $j=1,2,\ldots,m$, and $f_{ij}$
is the probability density of $R_{ij}$ for all $i,j$. Notice that the
boundedness condition in assumption \ref{asm:finity-of-repair-costs} assures
existence and finiteness of all these moments. However, assumption
\ref{asm:finity-of-repair-costs} yields even more: since $R\sim F(\vec p,\vec
q)$ is a random variable within $[0,\infty)$ (nonnegativity) and has finite
moments by the boundedness assumption, the distribution $F(\vec p,\vec q)$ is
\emph{uniquely} determined by the sequence of moments. This is made rigorous
by the following lemma:

\begin{lem}\label{lem:uniqueness}
Let two random variables $X,Y$ have their moment generating functions
$\mu_X(s), \mu_Y(s)$ exist within a neighborhood $U_\eps(0)$. Assume that
$\E{X^k}=\E{Y^k}$ for all $k\in\N$. Then $X$ and $Y$ have the same
distribution.
\end{lem}
The proof is merely a collection of well-known facts about moment generating
functions and the identity of their power-series expansions. For convenience
and completeness, we nevertheless give the proof in (almost) full detail.

\begin{proof}[of lemma \ref{lem:uniqueness}]
Let $Z$ be a general random variable. The finiteness of the moment-generating
function $\mu_Z$ within some open set $(-s_0,s_0)$ with $s_0>0$ yields
$\E{Z^k}=\mu_z^{(k)}(0)$ via the $k$-th order derivative of $\mu_Z$
\cite[Theorem 3.4.3]{Evans2004}. Furthermore, if the moment generating
function exists within $(-s_0,s_0)$, then it has a Taylor-series expansion
(cf. \cite[Sec.11.6.1]{Golberg1984}).
\begin{equation}\label{eqn:taylor-series}
  \mu_Z(s) = \sum_{k=0}^\infty \frac{\mu_Z^{(k)}(0)}{k!}s^k,\qquad\forall s\in(-s_0,s_0).
\end{equation}
Identity of moments between $X$ and $Y$ (the lemma's hypothesis) thus implies
the identity of the Taylor-series expansions of $\mu_X$ and $\mu_Y$ and in
turn the identity $\mu_X(s)=\mu_Y(s)$ on $(-s_0,s_0)$. This equation finally
implies that $X$ and $Y$ have the same distribution by the uniqueness theorem
of moment-generating functions \cite[Theorem 3.4.6]{Evans2004}.
\end{proof}

Lemma \ref{lem:uniqueness} is the permission to characterize random variables
only by their moment-sequence to uniquely pin-down the probability
distribution, i.e., we will hereafter write $m_R(k) := \E{R^k}$, and use
\begin{equation}\label{eqn:sequence-representation}
(m_R(k))_{k\in\N}, \quad \text{to represent the random variable}\quad  R\sim F(\vec p,\vec q).
\end{equation}

Let $\R^\infty$ denote the set of all sequences, on which we define a partial
ordering by virtue of the above characterization as follows: let $F_1 =
F(\vec p_1,\vec q_1), F_2 = F(\vec p_2,\vec q_2)$ be two distributions
defined by \eqref{eqn:game-outcome-distribution}. As a first try, we could
define a preference relation between two distributions $F_1,F_2$ by comparing
their moment sequences element-wise, i.e., we would prefer $F_1$ over $F_2$
if the respective moments satisfy $m_{R_1}(k)\leq m_{R_2}(k)$ for all $k$
whenever $R_1\sim F_1$ and $R_2\sim F_2$.

It must be stressed that without extra conditions, this ordering is at most a
partial one, since we could allow infinitely alternating values for the
moments in both sequences. To make the ordering total, we have to be specific
on which indices matter and which don't. The result will be a standard
ultrapower construction, so let $\UF$ denote an arbitrary ultrafilter.
Fortunately, the preference ordering by comparing moments elementwise is
ultimately independent of the particular ultrafilter in use. This is made
precise in theorem \ref{thm:ordering-invariance} that is implied by a simple
analysis of continuous distributions. We treat these first and discuss the
discrete case later, as all of our upcoming findings remain valid under the
discrete setting.

\paragraph{The Continuous Case:}

\begin{lem}\label{lem:ordering-invariance} For any two
probability distributions $F_1,F_2$ and associated random variables $R_1\sim
F_1, R_2\sim F_2$ that satisfy assumption \ref{asm:finity-of-repair-costs} on
the compact set $\Omega$ (covering both supports) and in the continuous case,
assume that $f_1(a)\neq f_2(a)$ at $a=\max\Omega$. Then, there is a $K\in\N$
so that either $[\forall k\geq K: m_{R_1}(k)\leq m_{R_2}(k)]$ or $[\forall
k\geq K: m_{R_1}(k)\geq m_{R_2}(k)]$.
\end{lem}
\begin{proof} Let $f_1, f_2$ denote the densities of the distributions $F_1,F_2$.
Fix the smallest $b^*>1$ so that $\Omega:=[1,b^*]$ covers both the supports
of $F_1$ and $F_2$. Consider the difference of the $k$-th moments, given by
\begin{align}
    \Delta(k) := \E{R_1^k}-\E{R_2^k} &= \int_{\Omega}x^k f_1(x)dx - \int_{\Omega}x^k f_2(x)dx\nonumber\\
    &=\int_{\Omega}x^k(f_1-f_2)(x)dx\label{eqn:difference-of-moments}.
\end{align}
Towards a lower bound to \eqref{eqn:difference-of-moments}, we distinguish
two cases:
\begin{enumerate}
  \item If $f_1(x)>f_2(x)$ for all $x\in\Omega$, then $(f_1-f_2)(x)>0$ and
      because $f_1,f_2$ are continuous, their difference attains a minimum
      $\lambda_2>0$ on the compact set $\Omega$. So, we can lower-bound
      \eqref{eqn:difference-of-moments} as
      $\Delta(k)\geq\lambda_2\int_{\Omega}x^kdx\To+\infty$, as
      $k\To\infty$.
  \item Otherwise, we look at the right end of the interval $\Omega$, and
      define
        \[
            a^* :=\inf\set{x\geq 1: f_1(x)>f_2(x)}.
        \]
      Without loss of generality, we may assume $a^*<b^*$. To see this,
note that if $f_1(b^*)\neq f_2(b^*)$, then the
      continuity of $f_1-f_2$ implies $f_1(x)\neq f_2(x)$ within a range
      $(b^*-\eps,b^*]$ for some $\eps>0$, and $a^*$ is the supremum of all
      these $\eps$. Otherwise, if $f_1(x)=f_2(x)$ on an entire interval
      $[b^*-\eps,b^*]$ for some $\eps>0$, then $f_1\not>f_2$ on $\Omega$
      (the opposite of the previous case) implies the existence of some
      $\xi<b^*$ so that $f_1(x)<f_2(x)$, and $a^*$ is the supremum of all
      these $\xi$  (see figure \ref{fig:density-lower-bounds} for an
      illustration). In case that $\xi=0$, we would have $f_1\geq f_2$ on
      $\Omega$, which is either trivial (as $\Delta(k)=0$ for all $k$ if
      $f_1=f_2$) or otherwise covered by the previous case.

      In either situation, we can fix a compact interval $[a,b]\subset
      (a^*,b^*)\subset[1,b^*]=\Omega$ and two constants
      $\lambda_1,\lambda_2>0$ (which exist because $f_1,f_2$ are bounded as
      being continuous on the compact set $\Omega$), so that the function
      \[
        \ell(k,x) := \left\{
                     \begin{array}{rl}
                       -\lambda_1x^k, & \hbox{if }1\leq x<a; \\
                       \lambda_2x^k, & \hbox{if }a\leq x\leq b.
                     \end{array}
                   \right.
      \]
        lower-bounds the difference of densities in
        \eqref{eqn:difference-of-moments} (see figure
        \ref{fig:density-lower-bounds}), and

    \begin{figure}
          \centering
            \includegraphics[scale=0.9]{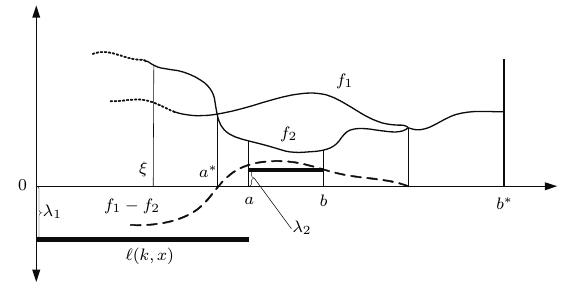}
          \caption{Lower-bounding the difference of densities}\label{fig:density-lower-bounds}
        \end{figure}

      \begin{align*}
        \Delta(k)=\int_1^{b^*}x^k(f_1-f_2)(x)dx &\geq \int_1^b\ell(x,k)dx \\
        &=-\lambda_1\int_1^ax^kdx + \lambda_2\int_a^bx^kdx\\
        &=-\frac{a^{k+1}}{k+1}(\lambda_1+\lambda_2)+\lambda_2\frac{b^{k+1}}{k+1}\To+\infty,
      \end{align*}
        as $k\To\infty$ due to $a<b$ and because $\lambda_1,\lambda_2$ are
        constants that depend only on $f_1,f_2$.

    In both cases, we conclude that, unless $f_1=f_2$, $\Delta(k)>0$ for
    sufficiently large $k\geq K$ where $K$ is finite.
\end{enumerate}

\end{proof}

\begin{thm}\label{thm:ordering-invariance} Let $\F$ be the set of
distributions that satisfy assumption \ref{asm:finity-of-repair-costs}.
Assume the elements of $\F$ to be represented by hyperreal numbers in
$\R^\infty\slash\UF$, where $\UF$ is any free ultrafilter. There exists a
total ordering on a dense subset of $\F$ that is independent of $\UF$.
\end{thm}
\begin{proof}
Let $F_1, F_2$ be two probability distributions, and let $R_1\sim F_1,
R_2\sim F_2$. Let $\Omega$ be the union of both supports of $R_1,R_2$, and
put $a=\max\Omega$. If $f_1(a)=f_2(a)$, then we can truncate the support and
the densities $f_1,f_2$ by some (small) $\eps>0$ such that $f_1(a-\eps)\neq
f_2(a-\eps)$. If no such $\eps$ would exist, the two densities are identical
and we have equivalence. Otherwise, by continuity, we can choose $\eps$
arbitrarily small to approximate $F_1$ and $F_2$ to any desired precision,
thus picking an element of a dense subset of $\F$ on which the conditions of
Lemma \ref{lem:ordering-invariance} hold. Call them $F_1,F_2$ again to ease
our notation. Lemma \ref{lem:ordering-invariance} assures the existence of
some $K\in\N$ so that $F_1\preceq F_2$ by $m_{R_1}(k)\leq m_{R_2}(k)$
whenever $k\geq K$. Let $L$ be the set of indices where $m_{R_1}(k)\leq
m_{R_2}(k)$, then complement set $\N\setminus L$ is finite (it has at most
$K-1$ elements). Let $\UF$ be an arbitrary ultrafilter. Since $\N\setminus L$
is finite, it cannot be contained in $\UF$ as $\UF$ is free. And since $\UF$
is an ultrafilter, it must contain the complement a set, unless it contains
the set itself. Hence, $L\in\UF$, and the claim follows. The converse case is
treated analogously.
\end{proof}

\begin{rem}\label{rem:responding-to-vincent}
If two distributions alternate infinitely often within a compact interval,
then the two would not compare by their tail masses as Lemma
\ref{lem:ordering-invariance}, without the hypothesis $f_1(a) \neq f_2(a)$.
Examples are constructible by ``squeezing'' alternating densities (cf.
Example \ref{exa:alternating-distributions}) into a compact range. However,
the assumption made by Lemma \ref{lem:ordering-invariance} is nonetheless
mild, since we can find truncations to make the distribution satisfy the
conditions, and choose the truncation to approximate the given densities up
to any precision that we desire. This argument extends to the order relation
even to approximations of distributions with unbounded support.
\end{rem}

Now, we can state our preference criterion on distributions on the quotient
space $\F \subset \R^{\infty}\slash\UF$, in which each probability
distribution of interest is represented by its sequence of moments. Thanks to
theorem \ref{thm:ordering-invariance}, there is no need to construct the
ultrafilter $\UF$ in order to well-define best responses, since two
distributions will compare in the same fashion under any admissible choice of
$\UF$.

\begin{defn}[Preference Relation over Probability Distributions]\label{def:preference}
Let $R_1,R_2$ be two random variables whose distributions $F_1, F_2$ satisfy
assumption \ref{asm:finity-of-repair-costs}. We \emph{prefer} $F_1$
\emph{over} $F_2$ relative to an ultrafilter $\UF$, written as
\begin{equation}\label{eqn:partial-ordering}
    F_1\preceq F_2:\iff \exists K\in\N\text{ s.t. }\forall k\geq K: m_{R_1}(k)\leq m_{R_2}(k)
\end{equation}
\emph{Strict preference} of $F_1$ over $F_2$ is denoted as
\[
    F_1\prec F_2:\iff \exists K\in\N\text{ s.t. }\forall k\geq K: m_{R_1}(k)<m_{R_2}(k)
\]
\end{defn}
Theorem \ref{thm:ordering-invariance} establishes this definition to be
compatible with (in the sense of being a continuation of) the ordering on the
hyperreals $\R^\infty\slash\UF$, being defined as $a\leq b$ iff $\set{i:
a_i\leq b_i}\in\UF$, when $a, b$ are represented by sequences
$(a_i)_{i\in\N}, (b_i)_{i\in\N}$. The more general version of Definition
\ref{def:preference} would thus put two random variables $X,Y$ into order $X
\preceq Y$ if and only if their representative moment-sequences satisfy
$(E(X^k))_{k\in\N}\leq (E(Y^k))_{k\in\N}$ when treated as hyperreal numbers.
The above results restrict the ordering relation to a dense subset of $\F$,
but have the appeal of a nicer interpretation (Theorem \ref{thm:tail-bounds})
and algorithmic decidability.

By virtue of the $\preceq$-relation, we can define an equivalence $\equiv$
between two distributions in the canonical way as
\begin{equation}\label{eqn:equivalence-relation}
    F_1\equiv F_2:\iff (F_1\preceq F_2)\land (F_2\preceq F_1).
\end{equation}

Within the quotient space $\F\subset\R^\infty\slash\UF$, we thus consider two
distributions as identical, if only a finite set of moments between them
mismatch. Observe that this \emph{does not} imply the identity of the
distribution functions themselves, unless actually all moments match.

The strict preference relation $\prec$ induces an ordering topology
$\mathcal{T}$ on $\F$, whose open sets are for any two distributions $F_1,
F_2$,
\[
    (F_1,F_2) := \set{F\in\F: F_1\prec F\prec F_2},
\]
and the topology is denoted as $\mathcal{T}=\set{(F_1,F_2)|F_1,F_2\in\F\text{ where
}F_1\prec F_2}$.

\paragraph{The Discrete Case:}
In situations where the game's payoffs are better modeled by discrete random
variables, say if a nominal scale (``low'', ``medium'', ``high'') or a
scoring scheme is used to express revenue, assumption
\ref{asm:finity-of-repair-costs} is too strong in the sense of prescribing a
continuous density where the model density is actually discrete.

Assumption \ref{asm:finity-of-repair-costs} covers discrete distributions
that possess a density w.r.t. the counting measure. The line of arguments as
used in the proof of Lemma \ref{lem:ordering-invariance} remains intact
without change, except for the obvious difference that $\Omega$ is a finite
(and hence discrete) set now. Likewise, all conclusions drawn from lemma
\ref{lem:ordering-invariance}, including theorem
\ref{thm:ordering-invariance}, as well as the definitions of ordering and
topology transfer without change.

\subsection{Comparing Discrete and Continuous Distributions}
The representation \eqref{eqn:sequence-representation} of distributions by
the sequence of their moments works even without assuming the density to be
continuous. Therefore, it elegantly lets us compare distributions of mixed
type, i.e., continuous vs. discrete distributions on a common basis.

It follows that we can -- without any changes to the framework -- compare
discrete to continuous distributions, or any two distributions of the same
type in terms of the $\preceq$-, $\prec$- and $\equiv$-relations. This
comparison is, obviously, only meaningful if the respective random variables
live in the same (metric) space. For example, it would be meaningless to
compare ordinal to numeric data.

\subsection{Comparing Deterministic to Random}
In certain occasions, the consequence of an action may result in perfectly
foreseeable effects, such as fines or similar. Such deterministic outcomes
can be modeled as degenerate distributions (point- or
Dirac-masses)\footnote{Note that the canonic embedding of the reals within
the hyperreals represents a number $a\in\R$ by the constant sequence
$(a,a,\ldots)$. Picking up this idea would be critically flawed in our
setting, as any such constant sequence would be preferred over any
probability distribution (whose moment sequence diverges and thus overshoots
$a$ inevitably and ultimately).}. These are singular and thus violate
assumption \ref{asm:finity-of-repair-costs}, since there is no density
function associated with them, unless one is willing to resort to generalized
functions; which we do not do in this report. Nevertheless, it is possible to
work out the representation in terms of moment sequences. If $X$ is a random
variable that deterministically takes on the constant value $a$ all the time,
then the respective moment sequence has elements $\E{X^k}=\E{a^k}=a^k$ for
all $k\in\N$. Given another non-degenerate distribution with density function
$f$, supported on $\Omega=[0,b]$, we can lower- or upper-bound the moments of
the respective random variable $Y$ by exponential functions in $k$, which can
straightforwardly $\preceq$-, $\equiv$- or $\prec$-compared to the
representative $(a^k)_{k\in\N}$ of the (deterministic) outcome $a\in\R$.
Algorithmic details will follow in part two of this research report.

\subsection{Extensions: Relaxing Assumption
\ref{asm:finity-of-repair-costs}}\label{sec:extensions} Risk management is
often required to handle or avoid extreme (catastrophic) events. The
respective statistical models are distributions with so-called ``heavy'',
``long'' or ``fat'' tails (exact definitions and distribution models will
follow in part two of this report). Extreme-value distributions such as the
Gumbel-distribution, or also the Cauchy-distribution (that is not an extreme
value model) are two natural examples that fall into the class of
distributions that assign unusually high likelihood to large outcomes (that
may be considered as catastrophic consequences of an action). In any case,
our assumption \ref{asm:finity-of-repair-costs} rules out such distributions
by requiring compact support. Even worse, the $\prec$-relation based on the
representation of a distribution by the sequence of its moments cannot be
extended to cover distributions with heavy tails, as those typically do not
have finite moments or moment-generating functions. Nevertheless, such
distributions are important tools in risk management.

Things are, however, not drastically restricted by assumption
\ref{asm:finity-of-repair-costs}, for at least two reasons: First,
compactness of the support is not necessary for all moments to exist, as the
Gaussian distribution has moments of all orders and is supported on the
entire real line (thus violating even two of the three conditions of
assumption \ref{asm:finity-of-repair-costs}). Still, it is characterized
entirely by its first two moments, and thus can easily be compared in terms
of the $\prec$-relation.

Second, and more importantly, any distribution with infinite support can be
approximated by a truncated distribution. Given a random variable $X$ with
distribution function $F$, then the \emph{truncated distribution} is the
distribution of $X$ conditional on $X$ falling into a finite range, i.e., the
truncated distribution function $\hat F$ gives the conditional likelihood
\[
    \hat F(x) = \Pr(X\leq x|a\leq X\leq b).
\]
Provided that $F$ has a density function $f$, the truncated density function
is
\[
    \hat f(x) = \left\{
                  \begin{array}{ll}
                    \frac{f(x)}{F(b)-F(a)}, & a\leq x\leq b; \\
                    0, & \hbox{otherwise.}
                  \end{array}
                \right.
\]
In other words, we simply crop the density $f$ outside the interval $[a,b]$
and re-scale the resulting function to become a probability distribution
again.

Since every distribution function $F$ is non-decreasing and satisfies
$\lim_{x\To\infty}F(x)=1$, any choice of $\delta>0$ admits a value $b$ such
that $F(b)>1-\delta$. Moreover, since our random variables are all
non-negative, we have $\lim_{x\To 0^+}F(x) = \lim_{x\To 0^+}\int_0^x f(x)dx =
0$, since $F$ is right-continuous. It follows that the truncated distribution
density for variables of interest in our setting simplifies to $\hat f(x) =
f(x) / F(b)$. Now, let us compare a distribution $F$ to its truncated version
$\hat F$ in terms of the probabilities that we would compute:
\begin{align*}
    \abs{F(x) - \hat F(x)} &= \abs{\int_0^x f(t)dt - \int_0^x f(t)/F(b)dt}\\
&= \big|\int_0^x f(t)\underbrace{\left(1-\frac 1{F(b)}\right)}_{<\eps}dt\big|
<\eps\int_0^\infty f(t)dt=\eps,
\end{align*}
for sufficiently large $b$, which depends on the chosen $\eps>0$ that
determines the quality of approximation. Conversely, can find always find a
truncated distribution $\hat F$ that approximates $F$ up to an arbitrary
precision $\eps>0$. This shows that restricting ourselves to distributions
with compact support, i.e., adopting assumption
\ref{asm:finity-of-repair-costs}, causes no more than a numerical error that
can be made as small as we wish.

More interestingly, we could attempt to play the same trick as before, and
characterize a distribution with fat, heavy or long tails by a sequence of
approximations to it, arising from better and better precisions $\eps\To 0$.
In that sense, we could hope to compare approximations rather than the true
density in an attempt to extend the preference and equivalence relations
$\preceq$ and $\equiv$ to distributions with fat, heavy or long tails.

Unfortunately, such hope is wrong, as a distribution is not uniquely
characterized by a general sequence of approximations (i.e., we cannot
formulate an equivalent to lemma \ref{lem:uniqueness}), and the outcome of a
comparison of approximations is not invariant to how the approximations are
chosen (i.e., there is also no alike for lemma
\ref{lem:ordering-invariance}). To see the latter, take the quantile function
$F^{-1}(\alpha)$ for a distribution $F$, and consider the tail quantiles
$\overline{F}^{-1}(\alpha) = F^{-1}(1-\alpha)$. Pick any sequence
$(\alpha_n)_{n\To\infty}$ with $\alpha_n\To 0$. Since
$\lim_{x\To\infty}F(x)=1$, the tail quantile sequence behaves like
$\overline{F}^{-1}(\alpha_n)\To \infty$, where the limit is independent of
the particular sequence $(\alpha_n)_{n\To\infty}$, but only the speed of
divergence is different for distinct sequences.

Now, let two distributions $F_1, F_2$ with infinite support be given.  Fix
two sequences $\alpha_n$ and $\omega_n$, both vanishing as $n\To\infty$, and
set
\begin{equation}\label{eqn:support-sequences}
  a_n := \overline{F}_1^{-1}(\alpha_n) \leq  b_n := \overline{F}_2^{-1}(\omega_n).
\end{equation}
Let us approximate $F_1$ by a sequences of truncated distributions $\hat
f_{1,n}$ with supports $[0,a_n]$ and let the sequence $\hat f_{2,n}$
approximate $f_2$ on $[0,b_n]$. Since $a_n\leq b_n$ for all $n$, the proof of
lemma \ref{lem:ordering-invariance} then implies that the approximations with
support $[0,a_n]$ is always strictly preferable to the distribution with
support $[0,b_n]$, thus $\hat f_{1,n}\preceq \hat f_{2,n}$. However, by
replacing the ``$\leq $'' by a ``$\geq$'' in \eqref{eqn:support-sequences},
we can construct approximations to $F_1, F_2$ whose supports exceed one
another in the reverse way, so that the approximations would always satisfy
$\hat f_{1,n}\succeq \hat f_{2,n}$. It follows that the sequence of
approximations \emph{cannot} be used to unambiguously compare distributions
with infinite support, unless we impose some constraints on the tails of the
distributions and the approximations. The next lemma assumes this situation
to simply not occur, which allows to give a \emph{sufficient} condition to
unambiguously extend strict preference in the way we wish.

\begin{lem}\label{lem:approximation-comparison}
Let $F_1,F_2$ be two distributions supported on the entire nonnegative real
half-line $\R^+$ with continuous densities $f_1,f_2$. Let $(a_n)_{n\in\N}$ be
an arbitrary sequence with $a_n\To\infty$ as $n\to\infty$, and let $\hat
f_{i,n}$ for $i=1,2$ be the truncated distribution $f_i$ supported on
$[0,a_n]$.

If there is a constant $c<1$ and a value $x_0\in\R$ such that $f_1(x)<c\cdot
f_2(x)$ for all $x\geq x_0$, then there is a number $N$ such that all
approximations $\hat f_{1,n}, \hat f_{2,n}$ satisfy $\hat f_{1,n}\prec \hat
f_{2,n}$ whenever $n\geq N$.
\end{lem}
\begin{proof}
Throughout the proof, let $i\in\set{1,2}$. The truncated distribution density
that approximates $f_i$ is $f_i(x)/(F_i(a_n)-F_i(0))$, where $[0,a_n]$ is the
common support of $n$-th approximation to $f_1, f_2$. By construction,
$a_{n,i}\To\infty$ as $n\To\infty$, and therefore $F_i(a_n)-F_i(0)\To 1$ for
$i=1,2$. Consequently,
\[
    Q_n = \frac{F_1(a_n)-F_1(0)}{F_2(a_n)-F_2(0)}\To 1,\quad\text{ as }n\To\infty,
\]
and there is an index $N$ such that $Q_n > c$ for all $n\geq N$. In turn,
\[
    f_2(x)\cdot Q_n > f_2(x)\cdot c> f_1(x),
\]
and by rearranging terms,
\begin{equation}\label{eqn:truncated-comparison}
    \frac{f_1(x)}{F_1(a_n)-F_1(0)} < \frac{f_2(x)}{F_2(a_n)-F_2(0)},
\end{equation}
for all $x\geq x_0$ and all $n\geq N$. The last inequality
\eqref{eqn:truncated-comparison} lets us compare the two approximations
easily by the same arguments as have been used in the proof of lemma
\ref{lem:ordering-invariance}, and the claim follows.
\end{proof}

By virtue of lemma \ref{lem:approximation-comparison}, we can extend the
strict preference relation to distributions that satisfy the hypothesis of
the lemma but need not have compact support anymore. Precisely, we would
strictly prefer one distribution over the other, if all truncated
approximations are ultimately preferable over one another.

\begin{defn}[Extended Preference Relation
$\prec$]\label{def:extended-strict-preference} Let $F_1, F_2$ be distribution
functions of nonnegative random variables that have infinite support and
continuous density functions $f_1, f_2$. We \emph{(strictly) prefer} $F_1$
\emph{over} $F_2$, denoted as $F_1\prec F_2$, if for every sequence
$a_n\to\infty$ there is an index $N$ so that the approximations $\hat
F_{i,n}$ for $i=1,2$ satisfy $\hat F_{1,n}\prec \hat F_{2,n}$ whenever $n\geq
N$.

The $\succ$-relation is defined alike, i.e., the ultimate preference of $F_2$
over $F_1$ on any sequence of approximations.
\end{defn}

Definition \ref{def:extended-strict-preference} is motivated by the above
arguments on comparability on common supports, and lemma
\ref{lem:approximation-comparison} provides us with a handy criterion to
decide the extended strict preference relation.


\begin{exa}\label{exa:evd-comparison}
It is a simple matter to verify that any two out of the three kinds of
extreme value distributions (Gumbel, Frechet, Weibull) satisfy the above
condition, thus are strictly preferable over one another, depending on their
particular parametrization. 
\end{exa}

Definition \ref{def:extended-strict-preference} can, however, not applied to
every pair of distributions, as the following example shows.
\begin{exa}\label{exa:alternating-distributions} Take the ``Poisson-like'' distributions with parameter $\lambda>0$,
\[
    f_1(k) \propto \left\{
               \begin{array}{ll}
                 \frac{\lambda^{k/2}}{(k/2)!}e^{-\lambda}, & \hbox{when $k$ is even;} \\
                 0, & \hbox{otherwise.}
               \end{array}
             \right.,\quad
    f_2(k) \propto \left\{
               \begin{array}{ll}
                 0, & \hbox{when $k$ is even};\\
\frac{\lambda^{(k-1)/2}}{((k-1)/2)!}e^{-\lambda}, & \hbox{otherwise} \\
               \end{array}
            \right.
\]
It is a simple matter to verify that no constant $c<1$ can ever make
$f_1<c\cdot f_2$ and that all moments exist. However, neither distribution is
preferable over the other, since finite approximations based on the sequence
$a_n:=n$ will yield alternatingly preferable approximations.
\end{exa}

An occasionally simpler condition that implies the hypothesis of definition
\ref{def:extended-strict-preference} is
\begin{equation}\label{eqn:strict-preference-alternative-criterion}
  \lim_{x\To\infty}\frac{f_1(x)}{f_2(x)} = 0.
\end{equation}
The reason is simple: if the condition of definition
\ref{def:extended-strict-preference} were violated, then there is an infinite
sequence $(x_n)_{n\in\N}$ for which $f_1(x_n)\geq c\cdot f_2(x_n)$ for all
$c<1$. In that case, there is a subsequence $(x_{n_k})_{k\in\N}$ for which
$\lim_{k\To\infty} f_1(x_{n_k})/f_2(x_{n_k})\geq c$. Letting $c\To 1$, we can
construct a further subsequence of $(x_{n_k})_{k\in\N}$ to exhibit that
$\limsup_{n\To\infty}(f_1(x_n)/f_2(x_n))=1$, so that condition
\eqref{eqn:strict-preference-alternative-criterion} would be refuted.

\begin{rem}\label{rem:strict-preference-not-extensible}
It must be emphasized that the above line of arguments does not provide us
with a mean to extend the $\preceq$- or $\equiv$-relations accordingly. For
example, an attempt to define $\preceq$ and $\equiv$ as above is obviously
doomed to failure, as asking for two densities $f_1, f_2$ to satisfy
$f_1(x)\leq c_1\cdot f_2(x)$ ultimately (note the intentional relaxation of
$<$ towards $\leq$), and $f_2(x)\leq c_2\cdot f_1(x)$ ultimately for two
constants $c_1,c_2<1$ is nonsense.
\end{rem}

A straightforward extension of $\preceq$ can be derived from (based on) the
conclusion of lemma \ref{lem:approximation-comparison}:
\begin{defn}\label{def:extended-preference} Let $F_1,F_2$ be two
distributions supported on the entire nonnegative real half-line $\R^+$ with
continuous densities $f_1,f_2$. Let $(a_n)_{n\in\N}$ be a diverging sequence
towards $\infty$, and let $\hat F_{i,n}$ for $i=1,2$ denote the density $F_i$
truncated to have support $[0,a_n]$. We define $F_1\preceq F_2$ if and only
if for every sequence $(a_n)_{n\in\N}$ there is some index $N$ so that $\hat
F_{1,n}\preceq \hat F_{2,n}$ for every $n\geq N$.
\end{defn}

More compactly and informally spoken, definition
\ref{def:extended-preference} demands preference on all approximations with
finite support except for at most finitely many exceptions near the origin.

Obviously, preference among distributions with finite support implies the
extended preference relation to hold in exactly the same way (since the
sequence of approximations will ultimately become constant when $a_n$
overshoots the bound of the support), so definition
\ref{def:extended-preference} extends the $\preceq$-relation in this sense.
This observation justifies our choice of definition
\ref{def:extended-preference} as a valid extension of $\preceq$ from
distributions with compact support to those with infinite support.

Unfortunately, the extended $\preceq$-relation is not as easy to decide as
(extended) strict preference and usually calls for computing the moment
sequences analytically to be able to compare them in the long run.

Nevertheless, example \ref{exa:evd-comparison} substantiates the expectation
that practically relevant distributions over $\R^+$ may indeed compare w.r.t.
$\prec$ or $\succ$ by definition \ref{def:extended-strict-preference} or
criterion \eqref{eqn:strict-preference-alternative-criterion}. While it is
easy to exhibit distributions with infinite support that $\prec$-compare in
the sense of definition \ref{def:extended-preference}, their practical
relevance or even occurrence is not guaranteed. Example
\ref{exa:extended-preference-examples}, however, shows that preference
relations are indeed non-empty if they are defined like
\eqref{eqn:partial-ordering} for distributions with infinite support.

\begin{exa}\label{exa:extended-preference-examples}
Let $X\sim F_1$ be a Gaussian random variable. It is well-known that any
Gaussian distribution has finite moments of all orders, so let us call the
resulting sequence $(a_n)_{n\in\N}$. Furthermore, let us construct another
sequence $(b_n)_{n\in\N}$ being identical to $(a_n)_{n\in\N}$, except for
finitely many indices $I=\set{i_1,i_2,\ldots,i_k}$ for which we choose
$b_j<a_j$ whenever $j\in I$ and $b_j:=a_j$ otherwise. It is easy to see that
the existence of the moment-generating function $\mu_X$ for $F_1$ implies the
existence of a moment-generating function $\mu_Y$ for a random variable
$Y\sim F_2$ that has moment sequence $(b_n)_{n\in\N}$ (since the power-series
$\mu_X$ dominates the series $\mu_Y$). However, it is equally clear that
$\mu_X\neq \mu_Y$. Thus, $F_1\neq F_2$ although $F_1\equiv F_2$, since the
mismatch is only on finitely many indices, and the complement set of these
must be in the ultrafilter.
\end{exa}

\todo{Idee fuer eine Ausdehnung der $\preceq$-Relation: verlange $\prec$ auf
allen endlichen Approximationen. (Pathologisches) Beispiel konstruierbar? Wie
wuerde ein Entscheidungskriterium aussehen?}

\subsection{Interpretation and Implications of Preferences}
Having defined
preferences among probability distributions, we now look at what $F_1\preceq
F_2$ actually means. A simple first impression is gained by considering cases
in which the first few moments match. For that sake, let $F_1, F_2$ be two
distributions for which no preference has been determined so far.
\begin{itemize}
  \item If $\mu_1=\E{F_1}<\mu_2=\E{F_2}$, then we prefer the distribution
      with smaller mean. That is, decisions that yield to less average risk
      would be $\preceq$-preferred.
  \item If the means are equal, then we prefer whatever distribution has
      smaller second moment (by virtue of Steiner's theorem). In other
      words, the preferred among $F_1, F_2$ would be the one with smaller
      variance, or otherwise said, the distribution whose outcome is ``more
      predictable'' in the sense of varying less around its mean.
  \item Upon equal mean and variance, the third moment would make the
      difference. This moment is the \emph{skewness}, and we would prefer
      the distribution for which $\E{X^3}$ is smaller, i.e., the
      distribution that ``leans more to the left''. The respective
      distribution would assign more likelihood to smaller outcomes, thus
      giving less risk.
\end{itemize}
We refrain here from extending the above intuitions to cover cases when
kurtosis tips the scale, as the physical meaning of this quantity is
debatable and no consensus among statisticians exists so far. Instead, we
give the following result that makes the above intuitions more explicit in
the sense of (under certain hypotheses) saying that:
\begin{quote}
\emph{If $F_1\preceq F_2$, then ``extreme events'' are less likely to occur
under $F_1$ than under $F_2$.}
\end{quote}
The rigorous version of this, which especially clarifies the adjective
``extreme'', is the following theorem:
\begin{thm}\label{thm:tail-bounds}
Let $X_1\sim F_1, X_2\sim F_2$, where $F_1, F_2$ satisfy assumption
\ref{asm:finity-of-repair-costs} on the joint support
$\Omega=\supp(F_1)\cup\supp(F_2)=[a,b]$. Furthermore, let $f_1,f_2$ be the
density functions on $\Omega$, and assume $f_1(b)\neq f_2(b)$. If $F_1\preceq
F_2$, then there exists a threshold $x_0\in\Omega$ so that for every $x\geq
x_0$, we have $\Pr(X_1>x)\leq \Pr(X_2>x)$.
\end{thm}
\begin{proof} Let $f_1, f_2$ be the density functions of $F_1, F_2$. Call
$\Omega=\supp(F_1)\cup\supp(F_2)=[0,a]$ the common support of both densities,
and take $\xi=\inf\set{x\in\Omega: f_1(x)=f_2(x)=0}$. Suppose there were an
$\eps>0$ so that $f_1>f_2$ on every interval $[\xi-\delta,\xi]$ whenever
$\delta<\eps$, i.e., $f_1$ would be larger than $f_2$ until both densities
vanish (notice that $f_1=f_2=0$ on the right of $\xi$). Then the proof of
lemma \ref{lem:ordering-invariance} delivers the argument by which we would
find a $K\in\N$ so that $\E{X_1^k}>\E{X_2^k}$ for every $k\geq K$, which
would contradict $F_1\preceq F_2$. Therefore, there must be a neighborhood
$[\xi-\delta,\xi]$ on which $f_1(x)\leq f_2(x)$ for all $x\in
[\xi-\delta,\xi]$. The claim follows immediately by setting $x_0=\xi-\delta$,
since taking $x\geq x_0$, we end up with $\int_x^\xi f_1(t)dt\leq
\int_{x}^\xi f_2(t)dt$, and for $i=1,2$ we have $\int_x^\xi f_i(t)dt =
\int_x^a f_i(t)dt = \Prob{X_i>x}$.
\end{proof}

Observe that this is compatible with the common goals of statistical risk
management \cite{McNeil2005} in other sectors, such as financial business:
the preference-relation $\preceq$ compares the tails of distributions, and
optimization w.r.t. $\preceq$ seeks to ``push'' the mass assigned by a
distribution towards lower damages. Essentially, we thus focus on large
deviations (damages), which intuitively makes sense, as small deviations from
the expected behavior may (most likely) be taken by the system's (designed)
natural resilience against distortions.

We close this section by giving a few examples showing graphically how
different distributions compare against each other.

\begin{exa}[different mean, same variance]
Consider two Gumbel-distributions $X\sim F_1=Gumbel(31.0063, 1.74346)$ and
$Y\sim F_2=Gumbel(32.0063, 1.74346)$, where a density for $Gumbel(a,b)$ is
given by
\[
f(x|a,b) = \frac 1 b e^{\frac{x-a}{b}-e^{\frac{x-a}{b}}},
\]
where $a\in\R$ and $b>0$ are the location and scale parameter.


\begin{figure}[h!]
          \centering
            \includegraphics[scale=0.9]{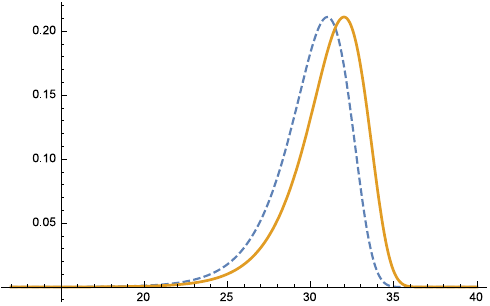}
          \caption{Comparing distributions with different means}\label{fig:different-mean-comparison}
        \end{figure}

Computations reveal that under the given parameters, the means are $\E{X}=30,
\E{Y}=31$ and $\Var{X}=\Var{Y}=5$. Figure \ref{fig:different-mean-comparison}
plots the respective densities of $F_1$ (dashed) and $F_2$ (solid line). The
respective moment sequences evaluate to
\begin{align*}
    \E{X^k}&=(30, 905, 27437.3, 835606, 2.55545\times 10^7,\ldots),\\
    \E{Y^k}&=(31, 966, 30243.3, 950906, 3.00162\times 10^7,\ldots),
\end{align*}
thus showing that $F_1\preceq F_2$. This is consistent with the intuition
that the preferred distribution gives \emph{less expected damage}.
\end{exa}

\begin{exa}[same mean, different variance]
Let us now consider two Gumbel-distributions $X\sim F_1=Gumbel(6.27294,
2.20532)$ and $Y\sim F_2=Gumbel(6.19073, 2.06288)$, for which $\E{X}=\E{Y}=5$
but $\Var{X}=8>\Var{Y}=7$.

\begin{figure}[h!]
  \centering
    \includegraphics[scale=0.9]{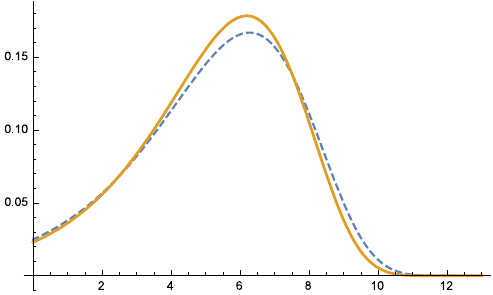}
  \caption{Comparing distributions with equal means but different variance}\label{fig:same-mean-comparison}
\end{figure}

Figure \ref{fig:same-mean-comparison} plots the respective densities of $F_1$
(dashed) and $F_2$ (solid line). The respective moment sequences evaluate to
\begin{align*}
    \E{X^k}&=(5, 33, 219.215, 1654.9, 11957.8,\ldots),\\
    \E{Y^k}&=(5, 32, 208.895, 1517.51, 10806.8,\ldots),
\end{align*}
thus showing that $F_2\preceq F_1$. This is consistent with the intuition
that among two actions leading to the same expected loss, the preferred one
would be one for which the variation around the mean is smaller; thus the
loss prediction is ``more stable''.
\end{exa}

\begin{exa}[different distributions, same mean and variance]
Let us now consider a situation in which the expected loss (first moment) and
variation around the mean (second moment) are equal, but the distributions
are different in terms of their shape. Specifically, let $X\sim
F_1=Gamma(260.345, 0.0373929)$ and $Y\sim Weibull(20,10)$, with densities as
follows:

\[
f_{\text{Gamma}}(x|a,b)=
\left \{
\begin{array}{cc}
 \frac{b^{-a} x^{a-1} e^{-\frac{x}{b}}}{\Gamma (a)}, & x>0; \\
 0, & \text{otherwise} \\
\end{array}
\right.
\]

\[
f_{\text{Weibull}}(x|a,b)   =\left\{
\begin{array}{cc}
 \frac{a e^{-\left(\frac{x}{b}\right)^a} \left(\frac{x}{b}\right)^{a-1}}{b}, & x>0; \\
 0, & \text{otherwise} \\
\end{array}
 \right.
\]

\begin{figure}[h!]
  \centering
    \includegraphics[scale=0.9]{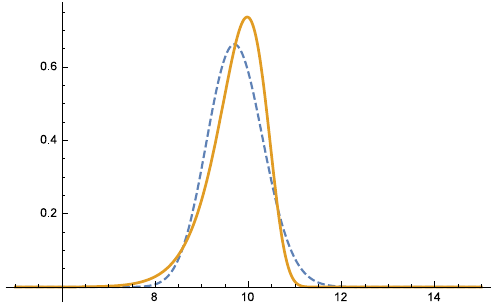}
  \caption{Comparing distributions with matching first two moments but different shapes}\label{fig:different-shape-comparison}
\end{figure}

 Figure \ref{fig:different-shape-comparison} plots the
respective densities of $F_1$ (dashed) and $F_2$ (solid line). The respective
moment sequences evaluate to
\begin{align*}
    \E{X^k}&=(9.73504, 95.1351, 933.259, 9190.01, 90839.7, \ldots),\\
    \E{Y^k}&=(9.73504, 95.1351, 933.041, 9181.69, 90640.2, \ldots),
\end{align*}
thus showing that $F_2\preceq F_1$. In this case, going with the distribution
that visually ```leans more towards lower damages'' would be flawed, since
$F_1$ nonetheless assigns larger likelihood to larger damages. The moment
sequence, on the contrary, unambiguously points out $F_2$ as the preferred
distribution. This illustrates Theorem \ref{thm:tail-bounds}.
\end{exa}

\section{Games with Uncertain Payoffs}\label{sec:generalized-game-theory}
Given a total ordering $\preceq$ on the set of actions as defined in section
\ref{sec:preferences}, we can go on lifting the remaining concepts and
results of game theory to our new setting. In particular, we will have to
investigate zero-sum competitions and Nash-equilibria in games whose payoffs
are probability distributions. Before, however, it pays to look at
arithmetics in our chosen subset of hyperreal numbers that represent our
payoff distributions. It turns out that things cannot be straightforwardly be
carried over, as we will illustrate in the next section.

\subsection{Arithmetic in $\F\subset\R^\infty\slash\UF$}
The space $\R^\infty\slash\UF$ is elsewhere known as the set of
\emph{hyperreal numbers}. Together with the ordering relation defined in the
same way as \eqref{eqn:partial-ordering}, and because $\UF$ is an
ultrafilter, $\F$ is actually a field, and in many ways behaves like the real
numbers. This is due to the ultrafilter acts in much the same way as a
maximal ideal, when the quotient structure is formed. For example, we can
soundly define $\min$- and $\max$-operators on $\F$. Furthermore, we can add
and subtract elements from $\F$ in the canonical way by applying the
respective operation pairwise on the sequences' elements. Likewise, we can
define an absolute-value function $\abs{x} = (\abs{x_i})_{i\in\N}$ on the
sequences, which naturally satisfies the triangle inequality because the
sequence's elements are from $\R$. However, we stress that the absolute value
does not induce a metric on $\F$ (even though it satisfies the necessary
conditions), as the absolute value under this definition is not real-valued.
This is one difference to the field $\R$.

A more important difference is the observation that any probability
distribution satisfying assumption \ref{asm:finity-of-repair-costs} can be
represented by an element in $\F$, but the converse is not true! For
instance, given $F\in\F$ as a representative of some probability
distribution, the element $(-F)$ as being the sequence of moments of $F$,
only with negative signs, does not represent a distribution (in general, and
specifically under assumption \ref{asm:finity-of-repair-costs}). Neither is
the sum of two moments necessarily the moment of some other probability
distribution. Finally, observe that the zero element $0=(0,0,\ldots)$ does
not define a proper probability distribution. Hence, the concept of a
``zero-sum game'' must be replaced by the (strategically equivalent) concept
of a constant-sum game, to properly define things. This issue will not be of
any particular importance in the following.

Our proofs will nevertheless heavily rely on the existence of a well-defined
ordering and arithmetic on the subset $\F$ of the hyperreals. The fact that
in lack of an explicit representation of $\UF$ we cannot do arbitrary
arithmetic somewhat limits the candidate algorithms to analyze the games and
compute equilibria and security strategies, however, this limitation is not
severe and can be overcome in our context of application.

\subsection{Continuity of $F(\vec p,\vec q)(r)$ in $(\vec p,\vec q)$}
The existence of Nash-equilibria crucially hinges on the continuity of payoff
functionals (in the classical setting). The existence of a topology on the
hyperreal set that we consider lets us soundly define continuity in terms of
the topology, but proving our payoff distribution function
\eqref{eqn:game-outcome-distribution} to be continuous is so far an open
issue, and this gap shall be closed now.

To establish continuity of the distribution-valued utility function $u(\vec
p,\vec q):=F(\vec p,\vec q)$ we have to show that any set in the topology
$\mathcal{T}$, i.e., any open set in $\F$, has a preimage under $u$ that is
open in $S_1\times S_2$ w.r.t. the product topology. The following lemma
establishes the important steps towards this conclusion by exploiting the
ordering and arithmetic within $\F$. Hereafter, we consider
$(\F\subset\R^\infty\slash\UF,\mathcal{T},\prec)$ as an ordered topological
space induced by an arbitrary ultrafilter $\UF$.

\begin{lem}\label{lem:continuity-of-scalar-products}
Let $r_1,\ldots,r_k\in\F$ for $k\geq 1$ be a set of fixed elements, and take
$\vec\alpha=(\alpha_1,\ldots,\alpha_k)\in\R^k$. If two elements $\ell,u\in\F$
bound the weighted sum $\ell\prec\sum_{i=1}^k\alpha_ir_i=\vec \alpha^T\vec
r\prec u$, then there is some strictly positive $\delta\in\R$ so that
$\ell\prec\vec{\wt\alpha}^T\vec r\prec u$ for every $\vec{\wt\alpha}$ within
a $\delta$-neighborhood of $\vec\alpha$ in $\R^k$.
\end{lem}
\begin{proof}
Define $\Delta :=\min\set{\vec\alpha^T\vec r-\ell, u-\vec\alpha^T\vec r}\succ
0$ and $r :=\max\set{r_1,\ldots,r_k}$. Suppose that we would modify all
weights $\alpha_i$ to $\alpha_i+\delta_i=\wt\alpha_i$. If so, then the
so-modified sum differs from the given one by $\abs{\vec{\wt\alpha}^T\vec
r-\vec\alpha^T\vec r}\leq\sum_{i=1}^k\abs{\delta_i} r_i\leq
r\cdot\sum_{i=1}^k\abs{\delta_i}$. Now, suppose that all
$\abs{\delta_i}\leq\delta$, then the change alters $\vec\alpha^T\vec r$ by a
magnitude of no more than $r\cdot\sum_{i=1}^k\delta_i\leq r\cdot k\cdot
\delta$. As $k$ and $r$ are fixed, we can choose $\delta$ sufficiently small
to satisfy $r\cdot k\cdot\delta\prec\Delta$, in which case we must have
$\abs{\vec{\wt\alpha}^T\vec r-\vec\alpha^T\vec r}<\Delta$, and therefore
$\ell\prec\vec{\wt\alpha}^T\vec r\prec u$ for any choice of $\wt\alpha$
within an $\delta$-neighborhood of $\vec\alpha$ in the maximum-norm on
$\R^k$.
\end{proof}

By virtue of lemma \ref{lem:continuity-of-scalar-products}, continuity of
$F(\vec p,\vec q)$ is easily implied by the continuity of the weights
$C_{\vec p,\vec q}(i,j)$ in $(\vec p,\vec q)$.

\begin{prop}\label{prop:continuity}
Let $i,j$ be integers and define the function $D_{ij}:S_1\times S_2\To\R$ as
$D_{ij}(\vec p,\vec q)=C_{\vec p,\vec q}(i,j)=\Pr_{\vec p,\vec q}(i,j)$. If
$D_{ij}$ is continuous and all $F_{ij}$ satisfy assumption
\ref{asm:finity-of-repair-costs}, then the mapping $F: S_1\times
S_2\To\F;\,\, (\vec p,\vec q)\mapsto\sum_{i,j}C_{\vec p,\vec q}(i,j)F_{ij}$
is continuous w.r.t. the product topology on $S_1\times S_2$ and the order
topology on $\F$.
\end{prop}
\begin{proof}
Without a metric on $\F$, we need to show that the preimage of every open set
in $\F$ under $F$ is open to prove that $F$ is continuous. For that sake, let
the open set $(\ell,u)\in\mathcal{T}$ be arbitrary and contain some point
$F(\vec p,\vec q)$ (which must exist, for otherwise, the set of preimages
would be empty). To ease notation, let us flatten the double-sum $\sum_{i,j}$
into an ordinary sum (say, by introducing a multiindex $\nu$) over $k=n\cdot
m$ elements, where $n,m$ are the limits in the original expression. Then, the
mapping takes the form $F(\vec p,\vec q)=\sum_{\nu=1}^k D_{\nu}(\vec p,\vec
q)F_\nu$. With the weights $\vec \alpha$ being defined by the individual
values of $D_\nu(\vec p,\vec q)=C_{\vec p,\vec q}(\nu)$, we can apply lemma
\ref{lem:continuity-of-scalar-products} to establish a bound $\delta>0$
within which we can arbitrarily alter the weights towards $\vec{\wt\alpha}$
without leaving the open set $(\ell,u)$. Since $C$ is continuous on compact
$S_1\times S_2$ it is also uniformly continuous, and we can fix a $\delta'>0$
so that $\norm{D_\nu(\vec p',\vec q')-D_\nu(\vec p,\vec q)}<\delta$ whenever
$\norm{(\vec p,\vec q)-(\vec p',\vec q')}<\delta'$, independently of the
particular point $(\vec p,\vec q)$. The sought pre-image of the open set
$(\ell,u)$ is thus the (infinite) union of open neighborhoods constructed in
the way described, and thus itself open.
\end{proof}

\subsection{Security Strategies and Zero-Sum Games}
Given the continuity and ordering of payoffs in games that reward players
with random variables, our next step is the definition of zero-sum games in
this context. We rephrase the standard definition of a zero-sum equilibrium
using the preference relation $\preceq$ in the straightforward fashion, by
defining a two-person game $\Gamma_0=(\set{1,2},\set{S_1,S_2},\set{F,-F})$ as
usual, but keeping the following in mind:
\begin{itemize}
  \item When $(a_i)_{i=1}^\infty$ defines the probability distribution $F$,
      then $-F$ is defined by $(-a_i)_{i=1}^\infty$, but not necessarily
      defines a valid probability distribution any more. To see this,
      simply recall that the Taylor-series \eqref{eqn:taylor-series} would
      upon all negative moments define a negative-valued function
      $\mu_X(s)<0$, which cannot be a moment-generating function, since
      $\mu_X(s)=\E{e^{sX}}\geq 0$ in any case.
  \item The sum $F+(-F)$ being computed in $\R^\infty\slash\UF$ is defined
      as the sequence that is constantly zero. Again, this does not define
      a probability distribution in the proper sense. However, strategic
      equivalence (as in the classical theory of games) tells the set of
      equilibria does not change if the payoffs of both players get a
      constant value added to them (in that case, the payoffs on either
      side are changed by the same value, leaving all inequalities intact).
      By the same token, we may think of constant-sum games, which avoid
      degenerate cases as above (where two distributions add up to
      something that is no longer a distribution).
\end{itemize}

The familiar equilibrium condition in a two-player zero-sum game $\Gamma$ can
be rephrased as follows: a strategy profile $(\vec p^*,\vec q^*)\in S_1\times
S_2$ is a \emph{(Nash-)equilibrium}, if for every $(\vec p,\vec q)\in
S_1\times S_2$,
\begin{equation}\label{eqn:equilibrium}
    F(\vec p,\vec q^*)\preceq F(\vec p^*,\vec q^*)\preceq F(\vec p^*,\vec q),
\end{equation}
i.e., any deviation from the optimal profile $(\vec p^*,\vec q^*)$ would
worsen the situation of either player (in either a zero- or constant-sum
competition).

Before security strategies can be defined properly, we need to assure
existence of equilibria profiles in our modified setting. In this regard,
Glicksberg's theorem, which generalizes Nash's original theorem, helps out:
\begin{thm}[see {\cite{Glicksberg1952} and \cite[Theorem
1.3]{Fudenberg1991}}]\label{thm:glicksberg} Consider a strategic-form game
whose strategy spaces $S_i$ are nonempty compact subsets of a metric space.
If the payoff functions are continuous w.r.t. the metric, then there exists a
Nash-equilibrium in mixed strategies.
\end{thm}
It is a simple matter to verify that
\begin{itemize}
  \item both sets $PS_1,PS_2$ are finite subsets of $\R$ (or $\R^d$ for
      $d>1$) and hence compact w.r.t. all norms on this Euclidian space,
      and
  \item the payoff functions are continuous by lemma \ref{prop:continuity},
\end{itemize}
so that a Nash-equilibrium in the sense of \eqref{eqn:equilibrium} exists by
theorem \ref{thm:glicksberg}.

Furthermore, any saddle-point satisfying \eqref{eqn:equilibrium} defines the
same payoff distribution in the sense of possibly defining a different
representative but in any case pinning down the same equivalence class of
distributions in $\R^\infty\slash\equiv$, where $\equiv$ is defined by
\eqref{eqn:equivalence-relation}. The proof is a restatement of Theorem 3.12
in \cite{Schlee2004}.
\begin{lem}\label{lem:saddle-value-invariance} Let a continuous function $F:PS_1\times PS_2\To\F$ be given, where
$PS_1\subseteq\R^{n_1}, PS_2\subseteq\R^{n_2}$. Furthermore, let $(\vec
p',\vec q')$ and $(\vec p^*,\vec q^*)$ be two different saddle-points . Then,
$(\vec p^*,\vec q')$ and $(\vec p',\vec q^*)$ are also saddle-points, and
\[
    F(\vec p',\vec q')\equiv F(\vec p^*,\vec q^*) \equiv  F(\vec p^*,\vec q') \equiv  F(\vec p',\vec q^*).
\]
\end{lem}
\begin{proof}
The proof is by direct checking of the saddle-point condition, i.e.,
\begin{align*}
F(\vec p^*,\vec q') &\preceq F(\vec p',\vec q')\preceq F(\vec p',\vec q^*)\preceq F(\vec p^*,\vec q^*)\preceq F(\vec p^*,\vec q')\\
\Rightarrow & F(\vec p',\vec q') \equiv  F(\vec p^*,\vec q^*) \equiv  F(\vec p',\vec q^*) \equiv  F(\vec p^*,\vec q').
\end{align*}
$(p',q^*)$ is saddle-point,
\[
    F(\vec p,\vec q^*)\preceq F(\vec p^*,\vec q^*)\equiv F(\vec p',\vec q^*) \equiv  F(\vec p',\vec q')\preceq F(\vec p',\vec q).
\]
The fact that $(\vec p^*,\vec q')$ is a saddle-point is proved analogously.
\end{proof}

Lemma \ref{lem:saddle-value-invariance} permits calling $v\equiv F(\vec
p^*,\vec q^*)$ \emph{the saddle-point value} of the zero-sum game $\Gamma_0$.
With this, we are ready to step forward towards defining \emph{security
strategies}.

For security strategies in the general case of two-person games with
arbitrary payoffs, let us denote the general game by
$\Gamma=(\set{1,2},\set{S_1,S_2},\set{F,G})$, in which player 1 has payoff
structure $F$, and player 2 has payoff structure $G$. Let $\Gamma_0$ denote
the associated zero-sum competition that -- adopting a worst-case assumption
-- substitutes an unknown $G$ by $(-F)$, i.e.,
$\Gamma_0=(\set{1,2},\set{S_1,S_2},\set{F,-F})$.

\begin{thm}\label{thm:scalar-case} Let $\Gamma$ be an arbitrary two-person game, and let $\Gamma_0$
be its associated zero-sum competition with equilibrium profile $(\vec
p^*,\vec q^*)$. Then, for every $(\vec p,\vec q)\in S_1\times S_2$, we have
\begin{equation}\label{eqn:equilibrium-bound}
    v\preceq F(\vec p,\vec q),
\end{equation}
and the strategy $\vec q^*$ achieves equality in
\eqref{eqn:equilibrium-bound}.
\end{thm}
\begin{proof}
Observe that the payoff $F(\vec p,\vec q)$ is the same for player 1 in both
games $\Gamma$ and $\Gamma_0$. So, if player 1 follows an equilibrium profile
$(\vec p^*,\vec q^*)$ of $\Gamma_0$, then the saddle-point condition
\eqref{eqn:equilibrium} yields
\begin{equation}\label{eqn:security-strategies-part1}
    v\equiv F(\vec p^*,\vec q^*)\preceq F(\vec p^*,\vec q),
\end{equation}
for every $\vec q$, with equality being achieved by $\vec q^*$, obviously.
Since player 2 will play to the best of its own benefit in $\Gamma$, call the
equilibrium profile in $\Gamma$ (which exists by theorem
\ref{thm:glicksberg}) $(\vec p,\vec q)$. In $\Gamma$, however, player 1
deviates by playing $\vec p^*$ thus increasing the payoff for player 2. Thus,
we can continue inequality \eqref{eqn:security-strategies-part1} on the right
side towards
\begin{equation}\label{eqn:security-strategies-part2}
F(\vec p^*,\vec q)\preceq F(\vec p,\vec q).
\end{equation}
The theorem is now immediate from expressions
\eqref{eqn:security-strategies-part1} and
\eqref{eqn:security-strategies-part2}.
\end{proof}

\section{Optimizing Multiple Security Goals}\label{sec:MGSS}
In case of multiple goals to be defended, we turn the two conclusions of
theorem \ref{thm:scalar-case} for scalar-valued games into two axioms on
vector-valued games. This leads to the following definition from
\cite{Rass2012}:
\begin{defn}[Multi-Goal Security Strategy with Assurance]\label{def:MGSS}\hfill\\
A strategy $\vec p^*\in S_1$ in a two-person multi-criteria game with
continuous payoff $\vec u_1:S_1\times S_2\To\F^d$ for the service provider
(player 1), is called a \emph{Multi-Goal Security Strategy with Assurance}
(MGSS) with \emph{assurance} $\vec v=(V_1,\ldots,V_d)\in\F^d$ if two criteria
are met:
\begin{description}
  \item[Axiom 1: Assurance] The values in $\vec v$ are the component-wise
      guaranteed payoff for player 1, i.e. for all components $i$, we have
  \begin{equation}\label{eqn:ds-req0}
    V_i\preceq u_1^{(i)}(\vec p^*,\vec q)\qquad\forall\vec q\in S_2,
  \end{equation}
  with equality being achieved by at least one choice $\vec q_i\in S_2$.
  \item[Axiom 2: Efficiency] At least one assurance becomes void if player
      1 deviates from $\vec p^*$ by playing $\vec p\neq \vec p^*$. In that
      case, some $\vec q_{\vec p}\in S_2$ exists (that depends on $\vec p$)
      such that
  \begin{equation}\label{eqn:ds-req2}
    \vec u_1(\vec p,\vec q_{\vec p})\preceq_1 \vec v.
  \end{equation}
\end{description}

\end{defn}

\subsection{Characterization and Existence of Security Strategies}

The existence of MGSS in the sense of definition \ref{def:MGSS} hinges on a
few basic facts about continuous real-valued functions. Fortunately, it turns
out that the only ingredient needed is uniform continuity of payoffs on
compact strategy spaces. The precise fact used to establish the existence of
multi-criteria security strategies is the following \cite{Rass2012}:
\begin{quote}
Let $u_1:PS_1\times PS_2\To\R^d$ be player 1's payoff function, and let it
be continuous. Since $PS_1\times PS_2$ is compact, given any $\eps>0$, we
can find a $\delta>0$ such that $\norm{\vec u_1(\vec x,\vec y)-\vec
u_1(\vec x',\vec y')}_\infty<\eps$, whenever $\norm{\vec x-\vec
y}_\infty<\delta$.
\end{quote}
This argument can be transferred easily to our setting, by a simple
inspection of the proofs of lemma \ref{lem:continuity-of-scalar-products} and
proposition \ref{prop:continuity}.

Proposition \ref{prop:continuity} tells that $F:S_1\times S_2\To\F$ is
continuous w.r.t. the topologies on $\R^{\abs{PS_1}\cdot\abs{PS_2}}$ and the
ordering topology on $\F$. So, let $(-\eps,+\eps)$ for $0\prec\eps\in\F$ be
an open interval, then we can find some real $\delta>0$ such that whenever
$\norm{(\vec p,\vec q)-(\vec p',\vec q')}_\infty<\delta$, we have $-\eps\prec
F(\vec p,\vec q)-F(\vec p',\vec q')\prec\eps$ by construction of $\delta$
(see the proofs of lemma \ref{lem:continuity-of-scalar-products} and
proposition \ref{prop:continuity}). More importantly, the $\delta$ is
constructed only from $\eps$ but is independent of $(\vec p,\vec q)$. Hence,
$F$ is indeed \emph{uniformly continuous}\footnote{The definition on
topological spaces (without relying on a metric) is the following: a function
$f:X\To Y$ is uniformly continuous, if for any neighborhood $B$ of zero in
$Y$, there is a neighborhood $A$ of zero in $X$ so that $x-y\in A$ implies
$f(x)-f(y)\in B$. This definition is satsified by our ``distribution-valued''
function $F:S_1\times S_2\To\F$.}. For vector-valued payoffs $\vec F:(\vec
p,\vec q)\mapsto (F^{(1)}(\vec p,\vec q),\ldots,F^{(d)}(\vec p,\vec q))$,
uniform continuity is inherited in the canonical way.

Furthermore, we need a proper replacement for the $\infty$-norm on $\R^d$,
which will work on elements $\vec x\in\F^d$. This replacement is
$\hrnorm{\vec x} = \hrnorm{(x_1,\ldots,x_d)} :=
\max\set{\abs{x_1},\ldots,\abs{x_d}}$ for $(x_1,\ldots,x_d)\in\F^d$, which
``resembles'' the $\infty$-norm on the real space. The slight difference in
the notation shall highlight the fact that $\hrnorm{\cdot}$ is technically
\emph{not} a norm, as it maps onto elements of $\F$ rather than real numbers.

Lemma \ref{lem:MGSS-optimality} is proved here for the sake of rigor, but is
the only part from \cite{Rass2012} that requires a reconsideration. The main
result needed here will be theorem \ref{thm:MGSS-characterization}, whose
proof will then rest on our version of lemma \ref{lem:MGSS-optimality}.

\begin{lem}\label{lem:MGSS-optimality}
Let $\Gamma$ be a multi-criteria game, and let $\vec p^*$ be a multi-goal
security provisioning strategy with assurance $\vec v$, assuming that one
exists. Then, no vector $\widetilde{\vec v}\prec\vec v$ is an assurance for
$\vec p^*$.
\end{lem}

\begin{proof}
Let $\widetilde{\vec v}\prec\vec v$, put $\eps:=\min_{1\leq i\leq
k}\set{v_i-\widetilde{v}_i}$ and observe that $\eps\succ0$. Since $\vec F$ is
uniformly continuous, a $\delta\succ0$ exists for which $\norm{(\vec p,\vec
q)-(\vec p',\vec q')}_\infty\prec\delta$ implies $\hrnorm{\vec F(\vec p,\vec
q)-\vec F(\vec p',\vec q')}\prec\frac \eps 2$.

Consider the mapping $\vec u_{\vec q}:S_1\To\R^k, \vec u_{\vec q}(\vec p) :=
\vec F(\vec p,\vec q)$, which is as well uniformly continuous on $S_1$ by the
same argument. So, $\norm{(\vec p^*,\vec q)-(\vec p',\vec
q)}_\infty=\norm{\vec p^*-\vec p'}_\infty\prec\delta$ implies $\hrnorm{\vec
u_{ \vec q}(\vec p^*)-\vec u_{\vec q}(\vec p')}=\max_{1\leq i\leq
k}\abs{F^{(i)}(\vec p^*,\vec q)-F^{(i)}(\vec p',\vec q)}\prec\frac \eps
2\quad\forall \vec q\in S_2$. It follows that $\abs{F^{(i)}(\vec p^*,\vec
q)-F^{(i)}(\vec p',\vec q)}\prec\frac{\eps}2$ for $i=1,\ldots,k$ and all
$\vec q\in S_2$, and consequently $\max_{\vec q\in S_2}\abs{F^{(i)}(\vec
p^*,\vec q)-F^{(i)}(\vec p',\vec q)}\prec\frac \eps 2$. Now, selecting any
$\vec p'\neq \vec p^*$ within an $\delta$-neigh\-borhood of $\vec p^*$, we
end up asserting $F^{(i)}(\vec p',\vec q)\succeq F^{(i)}(\vec p^*,\vec
q)-\frac \eps 2$ for every $i$ and $\vec q\in S_2$.

Using $F^{(i)}(\vec p^*,\vec q)\succeq v_i$, we can continue by saying that
$F^{(i)}(\vec p',\vec q)\succeq v_i-\frac \eps 2\succ v_i-\eps$. By
definition of $\eps$, we have $v_i-\widetilde{v}_i\succeq\eps$, so that
$F^{(i)}(\vec p',\vec q)\succ \widetilde{v}_i$ for all $i$, contradicting
\eqref{eqn:ds-req2} if $\widetilde{\vec v}$ were to be a valid assurance
vector.
\end{proof}

To compute MGSS, we apply a simple trick: we cast the two-person game in
which player 1 pursuits $d$ goals into a $(d+1)$-person game in which player
1 defends himself against $d$ adversaries, each of which refers to a single
security goal. The scenario is a ``one-against-all'' game, for which
numerical solution techniques (e.g., fictitious play) are known. This is
subject of upcoming companion work.

\begin{defn}[Auxiliary Game]\label{def:auxiliary-game}
Let a multiobjective game
\[
    \Gamma=(\set{1,2},\set{S_1,S_2},\set{\vec F_1,\vec
F_2})
\]
be given, where player 1 receives $d\geq 1$ outcomes through the (known)
payoff $\vec F_1=(F_1^{(1)},\ldots,F_1^{(d)})$. Assume $\vec F_2$ to be
unknown. We define the $(d+1)$-player multiobjective game
$\overline{\Gamma}=(N,S,H)$ as follows:
\begin{itemize}
  \item $N:=\set{0,1,\ldots,d}$, is the player set,
  \item $S := \set{S_1,S_2,\ldots,S_2}$ is the strategy multiset containing
      $d$ copies of $S_2$ (one for each opponent in $N$),
  \item the payoffs are
    \begin{itemize}
      \item vector-valued for player 0, who gets
        \[
            \overline{\vec F}_0(s_0,\ldots,s_d) :=
          (F_1^{(1)}(s_0,s_1),\ldots,F_1^{(d)}(s_0,s_d)),
        \]
      \item scalar for all opponents $i=1,2,\ldots,d$, receiving
          \[
            \overline{F}_i(s_0,\ldots,s_d) := -F_1^{(i)}(s_0,s_i).
          \]
    \end{itemize}
\end{itemize}
The game $\overline{\Gamma}$ is called the \emph{auxiliary game} for
$\Gamma$.
\end{defn}

\begin{thm}\label{thm:MGSS-characterization}
Let $\Gamma$ be a two-player multi-objective game with $d\geq 1$
distribution-valued payoffs. The situation $\vec p^*$ constitutes a network
provisioning strategy with assurance $\vec v$ for player 1 in the game
$\Gamma$, if and only if it is a Pareto-Nash equilibrium strategy for player
0 in the auxiliary $(d+1)$-player game $\overline{\Gamma}$ according to
definition \ref{def:auxiliary-game}.
\end{thm}
\begin{proof}[Sketch] The proof from \cite{Rass2012} transfers with obvious
changes to our setting, except for the above version of Lemma
\ref{lem:MGSS-optimality} being used in the last step.
\end{proof}

Theorem \ref{thm:MGSS-characterization} equates the set of multi-goal
security strategies to the set of Pareto-Nash equilibria in a conventional
game. Existence of such equilibria is assured by the following theorem:

\begin{thm}[{\cite{Lozovanu2005}}]\label{thm:pareto-nash-existence}
Let $\Gamma=(\set{1,\ldots,p},\set{S_1,\ldots,S_p},\set{\vec F_1,\ldots, \vec
F_p})$ be a $p$-player multiobjective game, where $S_1,\ldots,S_p$ are convex
compact sets and $\vec F_1,\ldots,\vec F_p$ represent vector-valued
continuous payoff functions (where payoff for player $i$ is composed from
$r_i\geq 1$ values). Moreover, let us assume that for every $i\in
\set{1,2,\ldots,p}$ each component
$F_i^{(k)}(s_1,s_2,\ldots,s_{i-1},s_i,s_{i+1},\ldots,s_p),
k\in\set{1,2,\ldots,r_i}$, of the vector function $\vec F_i$ represents a
concave function w.r.t. $s_i$ on $S_i$ for fixed $s_1,\ldots,s_{i-1},$
$s_{i+1},\ldots,s_p$. Then the multiobjective game $\Gamma$ has a Pareto-Nash
equilibrium.
\end{thm}

It is \emph{almost} straightforward to apply theorem
\ref{thm:pareto-nash-existence}, since almost all conditions have been
verified already: we have $p=2$ players, whose (vector-valued) payoffs are
continuous by proposition \ref{prop:continuity}, transferred canonically to
the vector-valued case (which means that player 0 in the auxiliary game has
$r_0=d$ payoffs, and every opponent $i=1,\ldots,d$ has $r_i=1$ payoff).
Likewise, the action spaces $PS_1,PS_2$ that we consider are finite subsets
of $\R$, and hence the simplex of discrete probability distributions
$S_1,S_2$ are compact and convex sets. However, it remains generally open
whether or not the payoff functions are concave. Under independent choices of
actions -- cf. section \ref{sec:dependent-actions} -- this is assured, and
theorem \ref{thm:pareto-nash-existence} applies. However, if the actions are
chosen interdependently, i.e., we have a nontrivial copula modeling the
interplay, concavity of the payoffs must be determined upon the explicit
structure of \eqref{eqn:copula-representation}.

\subsection{Relation to Bayesian
Decisions}\label{sec:relation-to-bayesian-choice} To embed the minimax-like
decision finding that we described in a Bayesian framework, recall that a
Bayesian decision is one that is optimal w.r.t. the a-posteriori
loss-distribution that incorporates all information. Informally spoken, such
decisions naturally give rise to minimax-decisions, if the loss-distribution
is the least favourable one. Our minimax approach, on the contrary, has the
opponent player 2 look exactly for this least favourable distribution, and
the zero-sum assumption then implies that a multi-goal security strategy in
the sense of \ref{def:MGSS} can be viewed as a Bayesian decision w.r.t. the
Pareto-Nash opponent-strategy in the ``zero-sum'' auxiliary game associated
with our multi-criteria competition (cf. theorem
\ref{thm:pareto-nash-existence}).

A full fledged treatment of Bayesian decision theory can be found in
\cite{Robert2001}. We abstain from transferring this framework to our
setting, as there appears to be no immediate benefit in doing so, due of the
inherent lack of information that risk management here. In other words, while
Bayesian decisions heavily rely on data, such data is not usually available
in the context of security and defenses against unknown attackers. Attacks
like eavesdropping are intrinsically unobservable (in most practically
relevant cases), and the consequences may be observed delayed and under
fuzzyness.

Summing up, minimax decision as we compute them here may indeed be
pessimistic relative to a ``more informed'' Bayesian decision. Under the
expected lack of information that risk management often suffers from,
however, it is nevertheless the best that we can do (theoretically).


\section{Compiling Quantitative Risk Measures}
The outcome of the game-theoretic analysis is in any case two-fold,
consisting of:
\begin{itemize}
  \item An optimal choice rule $\vec p^*$ over the set of actions $PS_1$,
      and
  \item An lower-bound distribution $V^*$ (or vector $\vec v$ if we
      optimize multiple goals as in section \ref{sec:MGSS}) for the random
      payoff that can be obtained from the game. This payoff is optimal in
      the sense of not being improvable without risking the existence of an
      attack strategy that causes greater damage than predicted by $V^*$.
      This bound is valid if and only if actions are drawn at random from
      $PS_1$ according to the distribution $\vec p^*$.
\end{itemize}

While the optimal action choice distribution $\vec p^*$ is easy to interpret,
compiling crisp risk figures from the payoff value $V^*$ (or a vector $\vec
v$ thereof) requires some more thoughts.

A common approach to risk quantisation is by the well-known ``rule-of-thumb''
\begin{equation}\label{eqn:risk-rule-of-thumb}
    \text{risk} = \text{(incident likelihood)}\times\text{(damage caused by the incident)},
\end{equation}
The beauty of this formula lies in its compatibility with any nominal or
numerical scale of likelihoods and damages, while at the same time, it enjoys
a rigorous mathematical fundament, part of which is game theory.

Indeed, formula \eqref{eqn:risk-rule-of-thumb} is essentially the
\emph{expected value of the loss-distribution} that is specified by the
damage potential of all known incidents, together with their likelihoods. The
distribution $V^*$ that we obtain from our analysis of games with
distribution-valued payoffs is much more general and thus informative: let
$v$ be the optimal distribution, then:
\begin{itemize}
  \item Formula \eqref{eqn:risk-rule-of-thumb} is merely the first moment
      of $V^*$, i.e.,
    \begin{align*}
\text{risk} &= \text{likelihood}\times\text{damage}=\E{R},\text{ when }R\sim V^*
    \end{align*}
where the last quantity is equal to $\E{V^*}=\E{R^1}(\vec p^*,\vec q^*)$
  that can be computed from equation \eqref{eqn:moment}. The missing value
  $\vec q^*$ is here exactly the optimal strategy for the attacker in the
  hypothetical zero-sum competition that is set up to compute the sought
  security strategy $\vec p^*$. In other words, the value $\vec q^*$ is a
  natural by-product of the computation of $\vec p^*$ and delivered
  together with it.
  \item Beyond the crisp result that formula \eqref{eqn:risk-rule-of-thumb}
      delivers, the distribution $V$ can be analyzed for higher moments
      too, such as \emph{variance} of the damage, or quantiles that would
      provide us with probabilistic risk bounds: for example, computing the
      $5\%$- and $95\%$-quantiles of $V$ gives two bounds within the damage
      will range with a $90\%$ likelihood. This may be another interesting
      figure for decision support, which cannot be obtained on the
      classical way via formula \eqref{eqn:risk-rule-of-thumb}.
\end{itemize}
If the results refer to a MGSS, then the above reasoning holds for every
component of the assurance vector $\vec v=(V_1^*,\ldots,V_d^*)$. That is,
risk figures can be computed independently for every aspect of interest.

\begin{rem}
The entries in the optimal attack strategy $\vec q^*$ are an optimal choice
rule over the set of attacker's actions $PS_2$. As such, they can be taken as
indicators to neuralgic spots in the infrastructure. However, it must be
emphasized that equilibria, and hence also security strategies, are
notoriously non-unique. Therefore, the indication by $\vec q^*$ is only one
among many other possible ones, and thus must not be used isolated from or as
a substitute for other/further information and expertise.
\end{rem}

\begin{rem} As an alternative quantity of interest, one may ask for the
expected maximum repair costs over a duration of (unchanging) infrastructure
provisioning $\vec p$ and risk situation $\vec q$. In adopting such an
approach, we can put $M_n :=\max\set{R_1,\ldots,R_n}$, but then find
      \begin{align*}
        \Prob{M_n\leq r}&=\Prob{R_1\leq r,R_2\leq r,\ldots,R_n\leq r}\\
        &=\Prob{R_1\leq r}\Prob{R_2\leq r}\cdots\Prob{R_n\leq r}
        = [(F(\vec p,\vec q))(r)]^n,
        \end{align*}
         if the repairs induce independent costs. Since $(F(\vec p,\vec
      q))(r)\leq 1$ for all $r$ by definition of a distribution function,
      we end up concluding that the long-run maximum is either zero or one,
      as $[(F(\vec p,\vec q))(r)]^n\To 0$ if $(F(\vec p,\vec q))(r)<1$, or
      remains $F(\vec p,\vec q)(r)=1$ otherwise.

      So the maximum is not as informative as we may hope under the
assumptions made. Nevertheless, modeling maxima is indeed the proper way to
control risk, and theorem \ref{thm:tail-bounds} fits our $\preceq$-relation
and framework quite well into these classical line of approaches.
\end{rem}

\section{Outlook}
So far, various practical issues have been left untouched, which will be
covered in companion work to this report. In particular, future discussions
will include:
\begin{itemize}
  \item Methods and models to capture extreme events (distributions
      commonly used in quantitative risk management)
  \item Methods and algorithms to compile payoff distributions from
      simulation or empirical data
  \item Algorithms to efficiently decide preference and equivalence among
      probability distributions
  \item Algorithms to numerically compute security strategies that account
      for the limited arithmetic that we can do in lack of an explicit
      model of the hyperreal structure that represents our distributions.
\end{itemize}

This report is meant to provide the theoretical fundament to build the
practical analysis methods that are described in follow-up work. In that
sequel to this research, issues of modeling extreme events and damage
distributions for a game-theoretic risk control will be discussed.

%% file: notation.tex
This section is mostly intended to refresh the reader's memory about some
basic but necessary concepts from calculus and probability theory that we
will use in the following to develop the theoretical groundwork. This
subsection can thus be safely skipped and may be consulted whenever necessary
to clarify details.

\paragraph{General Symbols:}
Sets, random variables and probability distribution functions are denoted as
upper-case letters like $X$ or $F$. Matrices and vectors are denoted as
bold-face upper- and lower-case letters, respectively. For finite sets, we
write $\abs{X}$ for the number of elements (cardinality). For real values
$\abs{a}$ denotes the absolute value of $a\in\R$. For arbitrary sets, the
symbol $X^k$ is the $k$-fold cartesian product of $X$; the set $X^\infty$
thus represents the collection of all infinite sequences $(a_1, a_2, a_3,
\ldots)$ with elements from $X$. We denote such a sequence as
$(a_n)_{n\in\N}$.

If $X$ is a random variable, then its probability distribution $F_X$ is told
by the notation $X\sim F_X$. Whenever this is clear from the context, we omit
the subscript to $F$ and write $X\sim F$ only. If $X$ lives on a discrete
set, then we call $X$ a \emph{discrete random variable}. Otherwise, if $X$
takes values in some infinite and uncountable set, such as $\R$, then we call
$X$ a \emph{continuous random variable}. For discrete distributions, we may
also use the vector $\vec p$ of probabilities of each event to denote the
distribution of the discrete variable $X$ as $X\sim\vec p$.

Calligraphic letters denote families (sets) of sets or probability
distributions, e.g., ultrafilters (defined below) are denoted as $\UF$, or
the family of all probability distributions being denoted as $\F$. The family
of subsets of a set $A$ is written as $\mathcal{P}(A)$ (the \emph{power-set}
of $A$). If $F\in\F$ is a probability distribution, then its density --
provided it exists -- is denoted by the respective lower-case letter $f$.

\paragraph{Topology and Norms:}
As our considerations in section \ref{sec:generalized-game-theory} will
heavily rely on concepts of continuity and compactness or openness of sets,
we briefly review the necessary concepts now.

A set $A$ is called \emph{open}, if for every $x\in A$ there is another open
set $B\subset A$ that contains $x$. The family $\mathcal{T}$ of all open sets
is characterized by the property of being closed under infinite union and
finite intersection. Such a set is called a \emph{topology}, and the set $X$
together with a topology $\mathcal{T}\subset \mathcal{P}(X)$ is called a
\emph{topological space}. An interval $A$ is called \emph{closed}, if its
complement (w.r.t. the space $X$) is open.

In $\R$, it can be shown that the open intervals are all of the form $\set{x:
a<x<b}$ for $a,b\in\R$ and $a<b$. We denote these intervals by $(a,b)$ and
the topology on $\R$ is the set containing all of them. Note the existence of
a total ordering $\leq$ on a space always induces the so-called
\emph{order-topology}, whose open sets are defined exactly the aforementioned
way. Closed intervals are denoted by square brackets, $[a,b]=\set{x:a\leq
x\leq b}$. An set $X\subset\R$ is called \emph{bounded}, if there are two
constants $a,b<\infty$ so that all $x\in X$ satisfy $a<x<b$. An subset of
$\R$ is called \emph{compact}, if and only if it is closed and bounded.

For $(X,d_X), (Y,d_Y)$ being two metric spaces, we call a function $f:X\To Y$
\emph{continuous}, if for every $x_0\in X$ and every $\eps>0$ there is some
$\delta>0$ for which $d_X(x_0,y)<\delta$ implies $d_Y(f(x_0),f(y))<\eps$. If
the condition holds with the same $\eps,\delta$ for every $x_0\in A\subseteq
X$, then we call $f$ \emph{uniformly continuous} on the set $A$. It can be
shown that if a function $f$ is continuous on a compact set $A$, then it is
also uniformly continuous on $A$ (in general, however, continuity does not
imply uniform continuity). In the following, we will need this result only on
functions mapping compact subsets of $\R$ to probability distributions (the
space that we consider there will be the set of hyperreal numbers, which has
a topology but -- unfortunately -- neither a metric nor a norm).

On a space $X$, we write $\norm{\vec x}$ to denote the norm of a vector $\vec
x$. One example is the $\infty$-norm on $\R^n$, which is $\norm{\vec
x}_\infty = \norm{(x_1,\ldots,x_n)}_\infty =
\max\set{\abs{x_1},\ldots,\abs{x_n}}$ for every $\vec x\in\R^n$. This induces
the metric $d_\infty(\vec x,\vec y)=\norm{\vec x-\vec y}_\infty$.

It can be shown that every metric space is also a topological space, but the
converse is not true in general. However, the above definition of continuity
is (on metric spaces) equivalent to saying that a function $f:X\To Y$ is
continuous, if and only if every open set in $Y\in\mathcal{T_Y}$ has an open
preimage $f^{-1}(B)\in\mathcal{T_X}$, when $\mathcal{T}_X,\mathcal{T}_Y$
denote the topologies on $X$ and $Y$, respectively. This characterization
works without metrics and will be used later to prove continuity of payoff
functions (see lemma \ref{lem:continuity-of-scalar-products} and proposition
\ref{prop:continuity}).

\paragraph{Probabilities and Moments:}
Let $A\subset\Omega$ be a subset of some measurable\footnote{We will not
require any further details on measurability or $\sigma$-algebras in this
report, so we spare details or an intuitive explanation of the necessary
concepts here.} space $\Omega$ and $F$ be a probability distribution
function. The \emph{probability measure} $\Pr_F(A)$ is the Lebesgue-Stieltjes
integral $\Pr_{F}(A)=\int_{A} dF$ (note that this general formulation covers
both, discrete and continuous random variables on the same formal ground).
Whenever the distribution is obvious from the context, we will omit the
subscript to the probability measure, and simply write $\Pr(A)$ as a
shorthand of $\Pr_F(A)$.

All probability distribution functions $F$ that we consider in this report
will have a density function $f$ associated with them. If so, then we call
the closure of the set $\set{x:f(x)>0}$ the \emph{support} of $F$, denoted as
$\supp(F)$. A \emph{degenerate distribution} on $\R$ is one that assigns
probability mass 1 to a finite number (or more generally, a null-set) of
points in $\R$. If $\Pr(A)=1$ for a singleton set $A=\set{a}$ and $a\in\R$,
then we call this degenerate distribution a \emph{point-mass} or a
\emph{Dirac-mass}. We stress that such distributions do not have a density
function associated with them in general\footnote{at least not within the
space of normal functions; the Dirac-mass is, however, an irregular
generalized function (a concept that we will not need here).}.

Many commonly used distributions have infinite support, such as the Gaussian
distribution. The density function can, however, be cut off outside a bounded
range $[a,b]$ and re-scaled to normalize to a probability distribution again.
This technique lets us approximate any probability distribution by one with
compact support (a technique that will come handy in section
\ref{sec:extensions}).

The \emph{expectation} of a random variable is (by the law of large numbers)
the long-run average of realizations, or more rigorously, defined as
$\E{X}=\int_\Omega x dF(x)dx$. The $k$-th moment of a distribution is the
expectation of $X^k$, which we is denoted and defined as $m_X(k) :=
\E{X^k}=\int_{\Omega}x^k dF(x)$, or also $\E{X^k}=\int_{\Omega}x^k f(x)dx$,
if $F$ has a density function $f$. Special roles play the first four moments
or values derived thereof. One prominent example is the \emph{variance}
$\text{Var}(X)=\E{X-\E{X}}^2=\E{X^2}-(\E{X}^2)$ (this formula is known as
Steiner's theorem). Of particular importance is the so-called
\emph{moment-generating function} $\mu_X(s) = \E{\exp(s\cdot X)}$, from which
the $k$-th moment can be computed by taking the $k$-th order derivative
evaluated at the origin, i.e., we have
$\E{X^k}=\left.\frac{d^k}{ds^k}\mu_X(s)\right|_{s=0}$. Moments do not
necessarily exist for all distributions (an example is the
Cauchy-distribution, for which all moments are infinite), but exist for all
distributions with compact support (that can be used to approximate every
other distribution up to arbitrary precision).

Multivariate distributions model vector-valued random variables. Their
distribution is denoted as $F_{X,Y}$, or shorthanded as $F$. For an
$n$-dimensional distribution, the respective density function is then of the
form $f(x_1,\ldots,x_n)$, having the integral $\int_{\R^n}
f(x_1,\ldots,x_n)d(x_1,\ldots,x_n)=1$. This joint distribution in particular
models the interplay between the (perhaps mutually dependent) random
variables $X_1,\ldots,X_n$. The \emph{marginal distribution} of any of the
variables $X_i$ (where $1\leq i\leq n$) is the unconditional distribution of
$X_i$ no matter what the other variables do. Its density function is obtained
by ``integrating out'' the other variables, i.e.,
\[
f_{X_i}(x_i)=\int_{\R^{n-1}} f(x_1,\ldots,x_{i-1},x_{i+1},\ldots,x_n)
d(x_1,\ldots,x_{i-1},x_{i+1},\ldots,x_n).\]
 A (marginal) distribution is
called \emph{uniform}, if its support is bounded and its density is a
constant. The joint probability of a multivariate event, i.e.,
multidimensional set $A$ w.r.t. to a multivariate distribution $F_{X,Y}$, is
denoted as $\Prob{F_{X,Y}}(A)$. That is, the distribution w.r.t. which the
probabilities are taken are given in the subscript, whenever this is useful
or necessary to make things clear.

A particular important class of distributions are \emph{copulas}. These are
multivariate probability distribution functions on the $n$-dimensional
hypercube $[0,1]^n$, for which all marginal distributions are uniform. The
importance of copula functions is due to Skl{\aa}rs theorem, which tells that
the joint distribution $F$ of the random vector $(X_1,\ldots,X_n)$ can be
expressed in terms of marginal distribution functions and a copula function
$C$ as $F(x_1,\ldots,x_n)=C(F_1(x_1),\ldots,F_n(x_n))$. So, for example,
independence of events can be modeled by the simple product copula
$C(x_1,\ldots,x_n)=x_1\cdot x_2\cdots x_n$. Many other classes of copula
functions and a comprehensive discussion of the topic as such can be found in
\cite{Nelsen1999}.

\paragraph{Convexity and Concavity:}
Let $V$ be a vector-space. We call a set $A\subset V$ \emph{convex}, if for
any two points $x,y\in A$, the entire line connecting $x$ to $y$ is also
contained in $A$. Let $f:\R\To\R$ be a function and take two values $a<b$.
The function $f$ is called \emph{convex}, if for every two values $x,y$, the
line between $f(a)$ and $f(b)$ upper-bounds $f$ between $a$ and $b$. More
formally, let $L_{a,b}(x)$ be the straight line from $f(a)$ to $f(b)$, then
$f$ is convex if $f(x)\leq L_{a,b}(x)$ for all $a\leq x\leq b$. A function
$f$ is called \emph{concave} if $(-f)$ is convex.

\paragraph{Hyperreal Numbers and Ultrafilters:}
Take the set $\R^\infty$ of infinite sequences over the real numbers $\R$. On
this set, we can define the arithmetic operations $+$ and $\cdot$ elementwise
on two sequences $a=(a_1,a_2,a_3,\ldots)=(a_n)_{n\in\N}\in\R^\infty$ and
$b=(b_1,b_2,b_3,\ldots)=(b_n)_{n\in\N}\in\R^\infty$ by setting $a+b =
(a_1+b_1,a_2+b_2,a_3+b_3,\ldots)$ and $a\cdot b=(a\cdot b_1,a\cdot
b_2,a_3\cdot b_3,\ldots)$. The ordering of the reals, however, cannot be
carried over in this way, as the sequences $a = (1,4,2,\ldots)$ and $b =
(2,1,4,\ldots)$ would satisfy $\leq$ on some components and $\geq$ on some
others. To fix this, we need to be specific on which indices matter for the
comparison, and which do not. The resulting family of index-sets can be
characterized to be a so-called \emph{free ultrafilter}, which is defined as
follows: a family $\UF\subseteq\mathcal{P}(\N)$ is called a \emph{filter}, if
the following three properties are satisfied:
\begin{itemize}
  \item $\emptyset\notin\UF$
  \item closed under supersets: $A\subseteq B$ and $A\in\UF$ implies
      $B\in\UF$
  \item closed under intersection: $A,B\in\UF$ implies $A\cap B\in\UF$
\end{itemize}
If, in addition, $A\notin\UF$ implies that $\UF$ contains the complement set
of $A$, then $\UF$ is called an \emph{ultrafilter}. A simple example of a
filter is the \emph{Fr\'{e}chet}-filter, which is the family $\set{A:
\text{the complement of $A$ is finite}}$. A filter is called \emph{free}, if
it contains no finite sets, or equivalently, if any filter that contains
$\UF$ is equal to $\UF$, i.e., $\UF$ is maximal w.r.t. the
$\supseteq$-relation. An application of Zorn's lemma to the semi-ordering
induced by $\supseteq$ shows the existence of free ultrafilter as being
$\supseteq$-maximal elements, extending the Fr\'{e}chet-filter.

An ultrafilter naturally induces an equivalence relation on $\R^\infty$ by
virtue of calling two sequences $a=(a_n)_{n\in\N}, b=(b_n)_{n\in\N}$
$\equiv_\UF$-equivalent, if and only if $\set{i: a_i = b_i}\in\UF$, i.e., the
set of indices on which $a$ and $b$ coincide belongs to $\UF$. The $\leq$-
and $\geq$-relations can be defined in exactly the same fashion. The family
of equivalence classes modulo $\UF$ makes up the set of \emph{hyperreal
numbers}, i.e., $^*\R = \set{[a]_{\UF}: a\in\R^\infty} = \R^\infty\slash\UF$,
where $[a]_\UF = \set{b\in\R^\infty: a\equiv_\UF b}$. In lack of an exact
model of $^*\R$ due to the non-constructive existence assurance of the
necessary free ultrafilter, unfortunately, we are unable to practically do
arbitrary arithmetic in $^*\R$. It will be shown (later and in part two of
this report) that everything that needs to be computed practically works
without $\UF$ being explicitly known.

\paragraph{Elements of Game Theory:}
Let $N=\set{1,2,\ldots,n}$ be a finite set. Let $PS_i$ be a finite set of
actions, and denote by $PS_{-i}$ the cartesian product $PS_{-i} = PS_1\times
PS_2\times\cdots\times PS_{i-1}\times PS_{i+1}\times \cdots\times PS_n$,
i.e., the product of all $PS_j$ \emph{excluding} $PS_i$.

A finite non-cooperative $n$-person game is a triple $(N,H,S)$, where the set
$H=\set{u_i:PS_i\times PS_{-i}\to\R:i\in N}$ contains all player's payoff
functions, and the family $S=\set{PS_1,\ldots,PS_n}$ comprises the strategy
sets of all players. The attribute \emph{finite} is given to the game if and
only if all $PS_i$ are finite. An \emph{equilibrium strategy} is an element
$x^*=(x_1^*,\ldots,x_n^*)\in\prod_{i=1}^n PS_i$, so that all $i\in N$ have
\begin{equation}\label{eqn:von-neumann-equilibrium}
    u_i(x_i^*,x_{-i}^*)\geq u_i(x_i,x_{-i}^*).
\end{equation}
That is, action $x_i^*$ gives the maximal outcome for the $i$-th player,
provided that all other players follow their individual equilibrium
strategies. Otherwise said, no player has an incentive to solely deviate from
$x_i^*$, as this would only worsen the revenue from the gameplay\footnote{It
should be mentioned that this not necessarily rules out benefits for
\emph{coalitions} of players upon jointly deviating from the equilibrium
strategy. This, however, is subject of cooperative game-theory, which we do
not discuss here any further.}. It is easy to construct examples in which no
such equilibrium strategy exists. To fix this, one usually considers
\emph{repetitions} of the gameplay, and defines the revenue for a player as
the \emph{long-run average} of all payoffs in each round. Technically, this
assures the existence of equilibrium strategies in all finite games (Nash's
theorem). We will implicitly rely on this possibility here too, while
explicitly looking at the outcome of the game in a single round. As this is
-- by our fundamental hypotheses in this report -- a random variable itself,
condition \eqref{eqn:von-neumann-equilibrium} can no longer be soundly
defined, as random variables are not canonically ordered. The core of this
work will therefore be on finding a substitute for the $\geq$-relation, so as
to properly restate \eqref{eqn:von-neumann-equilibrium} when random variables
appear on both sides of the inequality.
